\documentclass[12pt,reqno]{amsart}
\usepackage{amsmath}
\usepackage{amsthm}
\usepackage{latexsym}
\usepackage{amsfonts}
\usepackage{amssymb}
\usepackage{bm,dsfont}
\usepackage{color}
\usepackage{graphicx}
\usepackage{enumerate}

\usepackage{multirow}

\usepackage[normalem]{ulem}
\usepackage{mathrsfs}
\usepackage{mathtools}

\usepackage{comment}

\usepackage{array}
\usepackage{hhline}
\usepackage{hyperref}

\newtheorem{proposition}{Proposition}
\newtheorem{theorem}{Theorem}
\newtheorem{lemma}{Lemma}

\theoremstyle{definition}

\definecolor{gray}{rgb}{0.5,0.5,0.5}

\newcommand{\ip}[2]{\langle\,#1\,\vert\,#2\,\rangle} 
\newcommand{\kb}[2]{\vert\,#1\,\rangle\langle\,#2\,\vert} 
\newcommand{\no}[1]{\left\|\,#1\,\right\|} 
\newcommand{\tr}[1]{{\rm tr}\left[#1\right]} 
\newcommand{\ptr}[2]{{\rm tr}_{#1} #2} 
\newcommand{\id}{\mathds{1}} 

\newcommand{\Mo}{\mathsf{M}}
\newcommand{\Qo}{\mathsf{Q}}
\newcommand{\Po}{\mathsf{P}}

\newcommand{\No}{\mathsf{N}}
\newcommand{\Go}{\mathsf{G}}

\newcommand{\cond}{\,|\,}

\newcommand{\Lcorner}{{\mathchoice{\mathrel{\makebox[7pt][c]{\rule{.4pt}{7.5pt}\rule{5pt}{.4pt}}}}{\mathrel{\makebox[7pt][c]{\rule{.4pt}{7.5pt}\rule{5pt}{.4pt}}}}{\mathrel{\makebox[5.5pt][c]{\rule{.4pt}{5.25pt}\rule{3.5pt}{.4pt}}}}{\mathrel{\makebox[4pt][c]{\rule{.4pt}{3.75pt}\rule{2.5pt}{.4pt}}}}}}
\usepackage{booktabs}
\usepackage{tikz}
\usetikzlibrary{decorations.pathreplacing}
\usepackage[tight,footnotesize]{subfigure}
\definecolor{lightblue}{RGB}{179,179,255}
\definecolor{darkblue}{RGB}{126,126,255}
\definecolor{darkorange}{RGB}{255,135,0}
\definecolor{lightgreen}{RGB}{0,233,72}
\definecolor{darkgreen}{RGB}{0,153,51}
\newtheorem{result}{Result}
\newcommand{\rhoo}{\varrho} 
\newcommand{\phii}{\varphi} 
\newcommand{\Cb}{\mathbb{C}} 
\newcommand{\Fb}{\mathbb{F}} 
\newcommand{\Rb}{\mathbb{R}} 
\newcommand{\Zb}{\mathbb{Z}} 
\newcommand{\ca}[1]{\mathcal{#1}} 
\newcommand{\lc}[1]{\mathcal{L}(\Cb^{#1})} 
\newcommand{\elle}[1]{\mathcal{L}(#1)} 
\newcommand{\ff}{\mathcal{F}}
\newcommand{\ii}{\mathcal{J}}
\newcommand{\CR}[1]{\textrm{CR} #1} 
\newcommand{\ECR}[1]{\textrm{ECR} #1} 

\def\opaco{0.25}\def\evid{red}

\tikzset{declare function = {
ratan(\x) = rad(atan(\x)) ;
rsin(\x) = sin(deg(\x)) ;
rcos(\x) = cos(deg(\x)) ;
}}

\tikzset{pics/figQPII/.style args={#1/#2/#3/#4}{code={

\begin{scope}[scale={#4*(1-1/#1)}]

\begin{scope}[opacity={\ifnum#2=0 1 \else \opaco \fi}, transparency group]

\foreach \k/\area in {1/lightgreen,2/darkgreen}{

\fill[color=\area]
({-1/(#1^\k-1)},{-1/(#1^\k-1)}) -- ({-1/(#1^\k-1)},1) -- (1,1) -- (1,{-1/(#1^\k-1)}) -- ({-1/(#1^\k-1)},{-1/(#1^\k-1)}) ;

}

\foreach \k/\area in {1/lightblue,2/darkblue}{

\pgfmathparse{\k==1 && #1==2 ? 1 : 0}

\ifnum\pgfmathresult=1

\fill[color=\area] (0,0)circle(1);

\else

\pgfmathsetmacro{\thetaz}{pi-ratan(sqrt(#1^\k-1))}

\fill[domain={-\thetaz}:{\thetaz}, samples=40, color=\area]
({-1/(#1^\k-1)},{-1/(#1^\k-1)}) -- plot ({(#1^\k * rcos(\x + \thetaz) + 2 - #1^\k)/(2 * (1 - #1^\k))}, {(#1^\k * rcos(\x - \thetaz) + 2 - #1^\k)/(2*(1 - #1^\k))}) -- ({-1/(#1^\k-1)},{-1/(#1^\k-1)}) ;

\fi

}

\foreach \k/\curve in {1/darkorange,2/red}{

\pgfmathsetmacro{\cond}{\k==1 && #1==2 ? 1 : 0}

\ifnum\cond=1

\draw[color=\curve, thick] (0,0)circle(1);

\else

\pgfmathsetmacro{\thetaz}{pi-ratan(sqrt(#1^\k-1))}

\draw[thick, domain={-\thetaz}:{\thetaz}, samples=40, color=\curve]
plot ({(#1^\k * rcos(\x + \thetaz) + 2 - #1^\k)/(2 * (1 - #1^\k))}, {(#1^\k * rcos(\x - \thetaz) + 2 - #1^\k)/(2*(1 - #1^\k))}) ;

\draw[dashed, domain={\thetaz}:{2*pi-\thetaz}, samples=40]
plot ({(#1^\k * rcos(\x + \thetaz) + 2 - #1^\k)/(2 * (1 - #1^\k))}, {(#1^\k * rcos(\x - \thetaz) + 2 - #1^\k)/(2*(1 - #1^\k))}) ;

\fi

\draw[thick, dashed]
\ifnum\k=1 ({-1/(#1^2-1)},1) \else (1,1) \fi -- ({-1/(#1^\k-1)},1) -- ({-1/(#1^\k-1)},{-1/(#1^\k-1)}) -- (1,{-1/(#1^\k-1)}) -- \ifnum\k=1 (1,{-1/(#1^2-1)}) \else (1,1) \fi ;

\ifnum\cond=1

\draw (-1,0)node[anchor=north east]{\tiny $-1$};
\draw (0,-1)node[anchor=north east]{\tiny $-1$}node[anchor=north east]{\tiny $\phantom{-\frac{1}{\pgfmathparse{int(#1-1)}\pgfmathresult}}$};

\else

\fill[color=\curve] ({1/(1-#1^\k)},{1/(1-#1^\k)}) circle({0.4pt*(#1/(#1-1))}) ;

\draw ({-1/(#1^\k-1)},0)node[anchor=north east]{\tiny $-\frac{1}{\pgfmathparse{int(#1^\k-1)}\pgfmathresult}$};
\draw (0,{-1/(#1^\k-1)})node[anchor=north east]{\tiny $-\frac{1}{\pgfmathparse{int(#1^\k-1)}\pgfmathresult}$};

\fi

}

\pgfmathsetmacro{\offs}{0.13*#1/(#1-1)}

\draw[->]
({-\offs-1/(#1-1)},0) -- ({1+\offs},0)node[anchor=north]{\small $s$} ;
\draw[->]
(0,{-\offs-1/(#1-1)}) -- (0,{1+\offs})node[anchor=east]{\small $t$} ;

\draw (1,0)node[anchor=north west]{\tiny $1$};
\draw (0,1)node[anchor=south east]{\tiny $1$};

\draw
({-\offs-1/(#1-1)},{\offs+1/(#1-1)+(#1-3)/(#1-1)}) -- ({\offs+1/(#1-1)+(#1-3)/(#1-1)},{-\offs-1/(#1-1)})
;
\draw
({-(1+\offs)+(#1^2-3)/(#1^2-1)},{1+\offs}) -- ({1+\offs},{-(1+\offs)+(#1^2-3)/(#1^2-1)})
;

\end{scope}

\ifnum#2>0

\pgfmathparse{#1==2 && #3==1 ? 1 : 0}

\ifnum\pgfmathresult=1

\draw[ultra thick, color=\evid]
([shift=({deg(3 * pi / 2 - #2 * (pi / 2))}:1)]0,0) arc ({deg(3 * pi / 2 - #2 * (pi / 2))} : {deg(pi / 2 - #2 * (pi / 2))} : 1);

\else

\ifnum#2=3

\fill[color=\evid] ({1/(1-#1^#3)},{1/(1-#1^#3)}) circle({0.6pt*(#1/(#1-1))}) ;

\else

\pgfmathsetmacro{\thetaz}{pi-ratan(sqrt(#1^#3-1))}
\pgfmathsetmacro{\thetau}{(-1)^#2 * \thetaz + int(#2/2) * pi}

\draw[ultra thick, domain={\thetau}:{\thetau+pi}, samples=40, color=\evid]
plot ({(#1^#3 * rcos(\x + \thetaz) + 2 - #1^#3)/(2 * (1 - #1^#3))}, {(#1^#3 * rcos(\x - \thetaz) + 2 - #1^#3)/(2*(1 - #1^#3))}) ;

\fi

\fi

\fi

\end{scope}

}}}

\tikzset{pics/figQI/.style args={#1/#2/#3}{code={

\begin{scope}[scale={#3*(1-1/(#1^2))}]

\pgfmathsetmacro{\thetaz}{pi-ratan(sqrt(#1^2-1))}

\begin{scope}[opacity={\ifnum#2=0 1 \else \opaco \fi}, transparency group]

\foreach \k/\area in {#1/lightgreen,{#1^2}/darkgreen}{

\fill[color=\area]
({-1/(\k-1)},{-1/(#1^2-1)}) -- ({-1/(\k-1)},1) -- (1,1) -- (1,{-1/(#1^2-1)}) -- ({-1/(\k-1)},{-1/(#1^2-1)}) ;

}

\foreach \z/\area in {-1/lightblue,1/darkblue}{

\pgfmathparse{\z==-1 && #1>2 ? 1 : 0}

\ifnum\pgfmathresult=1

\fill[domain={(#1 - 2) / (2 * (#1 - 1))}:1, samples=100, smooth, color=lightblue]
({- 1 / (#1 - 1)},{- 1 / (#1^2 - 1)}) -- plot ( {(1 / #1^2) * (- 2 * (1 - \x) - 2 * sqrt((1 - #1^2) * (\x)^2 + (#1^2 - 2) * \x + 1))}, \x ) -- (0,1) -- (0,{- 1 / (#1^2 - 1)}) -- cycle ;

\else

\fill[domain={-\thetaz}:{\thetaz}, samples=40, color=\area]
({\z * (- 1 / (#1^2 -1 ))},{- 1 / (#1^2 - 1)}) -- plot ({\z * (#1^2 * rcos(\x + \thetaz) + 2 - #1^2) / (2 * (1 - #1^2))}, {(#1^2 * rcos(\x - \thetaz) + 2 - #1^2) / (2 * (1 - #1^2))}) -- ({\z * (- 1 / (#1^2 - 1))},{- 1 / (#1^2 - 1)}) ;

\fi

}

\foreach \z in {-1,1}{

\pgfmathparse{\z==-1 && #1>2 ? 1 : 0}

\ifnum\pgfmathresult=1

\draw[thick, domain={(#1 - 2) / (2 * (#1 - 1))}:1, samples=100, smooth, red]
({- 1 / (#1 - 1)},{- 1 / (#1^2 - 1)}) -- plot ( {(1 / #1^2) * (- 2 * (1 - \x) - 2 * sqrt((1 - #1^2) * (\x)^2 + (#1^2 - 2) * \x + 1))}, \x ) -- (0,1) ;

\draw[domain={- 1 / (#1^2 - 1)}:1, samples=100, smooth, dashed]
plot ( {(1 / #1^2) * (- 2 * (1 - \x) + 2 * sqrt((1 - #1^2) * (\x)^2 + (#1^2 - 2) * \x + 1))}, \x ) -- (0,1) ;

\draw[domain={- 1 / (#1^2 - 1)}:{(#1 - 2) / (2 * (#1 - 1))}, samples=100, smooth, dashed]
plot ( {(1 / #1^2) * (- 2 * (1 - \x) - 2 * sqrt((1 - #1^2) * (\x)^2 + (#1^2 - 2) * \x + 1))}, \x ) ;

\else

\draw[thick, domain={\thetaz-pi}:{\thetaz}, samples=40, color=red]
plot ({\z * (#1^2 * rcos(\x + \thetaz) + 2 - #1^2) / (2 * (1 - #1^2))}, {(#1^2 * rcos(\x - \thetaz) + 2 - #1^2) / (2 * (1 - #1^2))}) ;

\draw[dashed, domain={\thetaz}:{\thetaz + pi}, samples=40]
plot ({\z * (#1^2 * rcos(\x + \thetaz) + 2 - #1^2) / (2 * (1 - #1^2))}, {(#1^2 * rcos(\x - \thetaz) + 2 - #1^2) / (2 * (1 - #1^2))}) ;

\fi

}

\foreach \k in {1,2}{

\draw[thick, dashed]
\ifnum\k=1 ({-1/(#1^2-1)},1) \else (1,1) \fi -- ({-1/(#1^\k-1)},1) -- ({-1/(#1^\k-1)},{-1/(#1^2-1)}) -- \ifnum\k=1 ({-1/(#1^2-1)},{-1/(#1^2-1)}) \else (1,{-1/(#1^2-1)}) -- (1,1) \fi ;

\draw ({-1/(#1^\k-1)},0)node[anchor=north east]{\tiny $- \pgfmathparse{\k==1 && #1==2 ? 1 : 0} \ifnum\pgfmathresult=1 1 \else \frac{1}{\pgfmathparse{int(#1^\k-1)}\pgfmathresult} \fi$};

\ifnum\k=2

\draw (0,{-1/(#1^2-1)})node[anchor=north east]{\tiny $-\frac{1}{\pgfmathparse{int(#1^2-1)}\pgfmathresult}$};

\fi

}

\pgfmathsetmacro{\offs}{0.18*(#1^2)/(#1^2-1)}

\draw[->]
({-\offs-1/(#1-1)},0) -- ({1+\offs},0)node[anchor=north]{\small $s$} ;
\draw[->]
(0,{-\offs-1/(#1^2-1)}) -- (0,{1+\offs})node[anchor=east]{\small $t$} ;

\draw (1,0)node[anchor=north west]{\tiny $1$};
\draw (0,1)node[anchor=south east]{\tiny $1$};

\draw
({-(1+\offs)+(#1^2-3)/(#1^2-1)},{1+\offs}) -- ({\offs+1/(#1^2-1)+(#1^2-3)/(#1^2-1)},{-\offs-1/(#1^2-1)});

\ifnum#1=2

\draw
({(2/#1^2)*(-\offs-1/(#1^2-1)-1)},{-\offs-1/(#1^2-1)}) -- ({(2/#1^2)*((1+\offs)-1)},{1+\offs});

\else

\draw
({-\offs-1/(#1-1)},{(#1/2)*(-\offs-1/(#1-1))+1}) -- ({(2/#1)*((1+\offs)-1)},{1+\offs});

\fill[color=red] ({1/(1-#1)},{1/(1-#1^2)}) circle({0.4pt*(#1/(#1-1))}) ;

\fi

\end{scope}

\pgfmathparse{#1>2 && #2==3 ? 1 : 0}

\ifnum\pgfmathresult=1

\draw[ultra thick, domain={- 1 / (#1^2 - 1)}:1, samples=100, smooth, color=\evid]
plot ( {(1 / #1^2) * (- 2 * (1 - \x) - 2 * sqrt((1 - #1^2) * (\x)^2 + (#1^2 - 2) * \x + 1))}, \x ) -- (0,1) ;

\else

\ifnum#2=4

\fill[color=\evid] ({1/(1-#1)},{1/(1-#1^2)}) circle({0.6pt*(#1/(#1-1))}) ;

\else

\ifnum#2>0

\def\z{\ifnum#2=3 -1 \else 1\fi}
\def\thetau{\ifnum#2=1 -\thetaz \else \thetaz+pi\fi}

\draw[ultra thick, domain={\thetau}:{\thetau+pi}, samples=40, color=\evid]
plot ({\z * (#1^2 * rcos(\x + \thetaz) + 2 - #1^2) / (2 * (1 - #1^2))}, {(#1^2 * rcos(\x - \thetaz) + 2 - #1^2) / (2 * (1 - #1^2))}) ;

\fi

\fi

\fi

\end{scope}

}}}

\begin{document}

\title[]{Postponing the Choice: Advantage of deferred measurements in quantum information processing}

\author[C.Carmeli]{Claudio Carmeli}
\address{Claudio Carmeli, DIME , Universit\`a di Genova, via Balbi 5, 16126 Genova, Italy\\
INFN Sezione di Genova, Via Dodecaneso 33, 16146 Genova, Italy}
\email[C.Carmeli]{claudio.carmeli@unige.it}
\author[T.Heinosaari]{Teiko Heinosaari}
\address{Teiko Heinosaari, Faculty of Information Technology, University of Jyv\"askyl\" a, Finland}
\email[T.Heinosaari]{teiko.heinosaari@jyu}
\author[A.Toigo]{Alessandro Toigo}
\address{Alessandro Toigo, Dipartimento di Matematica, Politecnico di Milano, Piazza Leonardo da Vinci 32, 20133 Milano, Italy \\
Istituto Nazionale di Fisica Nucleare, Sezione di Milano, Via Celoria 16, 20133 Milano, Italy}
\email[A.Toigo]{alessandro.toigo@polimi.it}

\begin{abstract}
Simultaneously implementing two arbitrary quantum measurements on the same system is impossible. The consequence of this limitation is that selecting one measurement actively excludes other possibilities. Two incompatible choices can then be forced together only at the cost of adding enough noise to the measurements.
An intriguing alternative is to postpone the choice, or part of it, until a later stage. We explore the advantages of this deferred decision-making and discover that the benefits critically depends on the assumptions about the forthcoming choice. In certain scenarios postponing the choice introduces no additional cost, while in others partial postponement can be effectively the same as full postponement.
\end{abstract}

\maketitle

\section{Introduction}\label{sec:intro}

The domain of quantum information revolves around the utilization of quantum systems to encode and convey information.
The necessary step in reading out the information written in the state of a quantum system is to carry out a measurement, which in the usual case is a sharp measurement corresponding to a Hilbert space orthonormal basis.
Quantum theory forbids simultaneous measurements of more than one basis, and for this reason one is forced to select a particular basis for readout.
Nonetheless, it is possible to measure more than one basis simultaneously by allowing noise in the measurement outcome distribution, which means that the uncertainty in measurement outcomes increases \cite{Busch86}.
These kind of approximate joint measurements have been extensively studied and the trade-off in noise levels has been characterized in several important cases, including the case of two mutually unbiased bases \cite{CaHeTo12,UoLuMoHe16,CaHeTo19,CaCaTo19,DeSkFrBr19}.

When performing a joint measurement, one has to decide beforehand which measurements will be done simultaneously, and then devise the joint measurement setup accordingly.
An alternative approach is to duplicate the initial state of the system and perform measurements on the duplicates.
Since quantum theory forbids perfect cloning of an unknown state, one is left with two approximate copies \cite{ScIbGiAc05}. 
The resulting noise in measurements is more than in optimal approximate joint measurements \cite{DaMaSa01,HeMiZi16}. 
That is not a surprise, as approximate cloning allows any choice of measurements, whereas a joint measurement is tailored for a specific pair of measurements only.
The advantage of approximate cloning is that the target measurements (e.g.~two bases) can be decided later.

There is a scenario between these two approaches in which one of the measurements is fixed beforehand but the choice of the other measurement is postponed. 
The first measurement necessarily disturbs the state of the system in some way, nevertheless, it can be aimed to perform in a manner that is not totally destructive and that is neutral with respect to all possible choices of the second measurement.
This approach gives an advantage over fixed joint measurements, as the other measurement can be decided later, for instance, after one has obtained information by means of the first measurement. 
However, it is not as universal as the approach based on approximate cloning, since one of the measurements is fixed and only the second measurement can be chosen freely.
In this way, the third scenario lies between the two earlier approaches, and one expects that there is a trade-off between freedom of choice and optimal noise levels.
The scenario has natural variations, as we can limit the choice of the second measurement from all possible choices to some more restricted class of choices.
The question is then to find the optimal way to perform the first measurement so that the noise trade-off is as favorable as possible.

\begin{figure}
\hfil\fbox{\begin{minipage}{\dimexpr \textwidth-2\fboxsep-2\fboxrule}
\vspace{0.3cm}
\centering
\subfigure[]{%
    \includegraphics[width=0.48\textwidth]{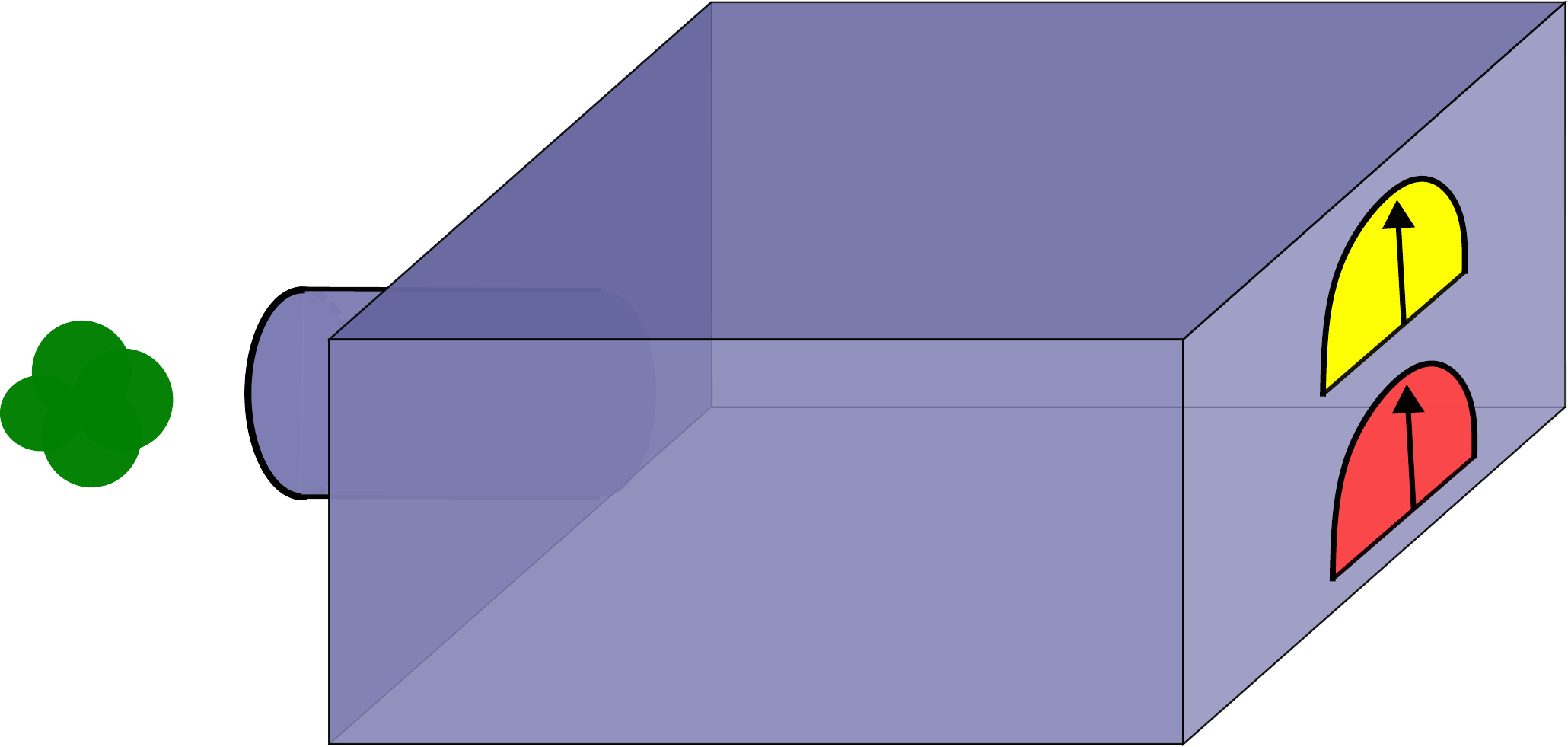}}\\ 
\subfigure[]{%
    \includegraphics[width=0.8\textwidth]{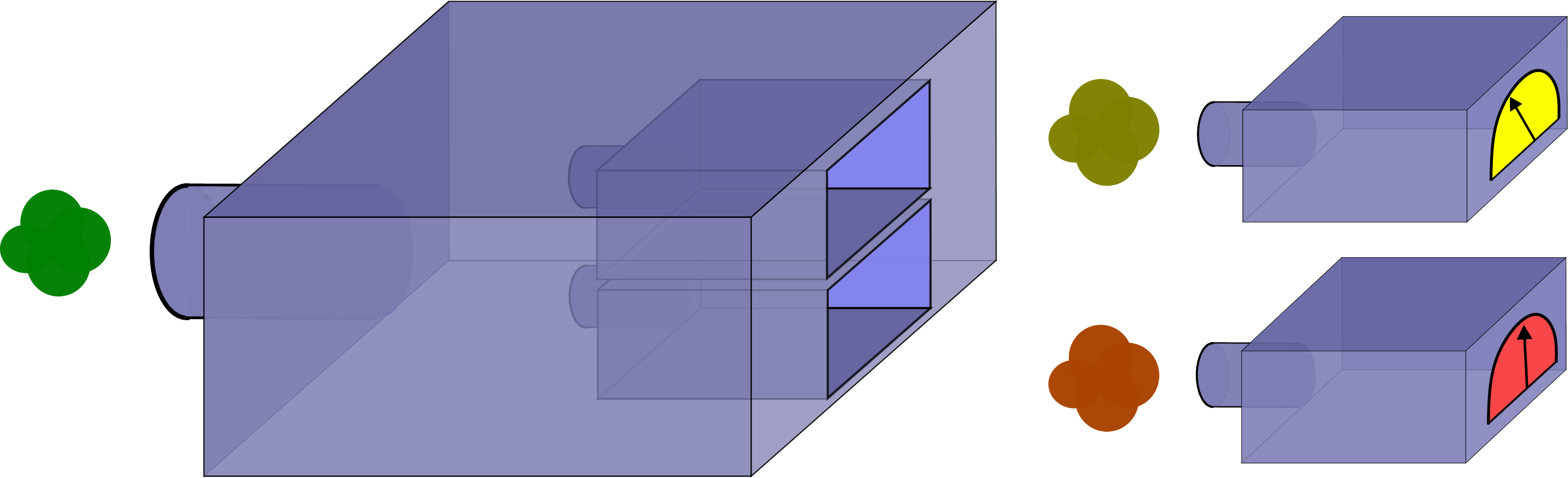}}
    \subfigure[]{%
    \includegraphics[width=0.85\textwidth]{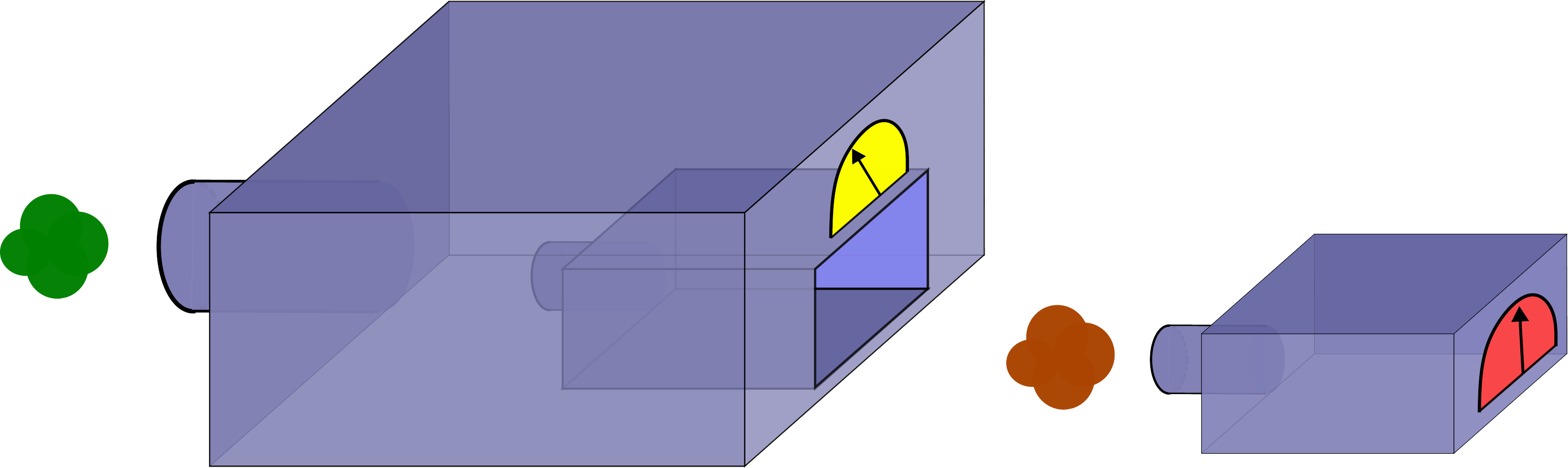}}
\end{minipage}}
\caption{The three different scenarios compared in the present investigation: (a) the target meters are fixed in advance and their approximate joint meter is tailored accordingly; (b) neither of the target meters is fixed in advance, and each of them is measured in a different approximate clone of the original state; (c) only the first target meter is fixed in advance, while the second one is decided after a measurement of the first has been performed.}\label{fig:3scenarios}
\end{figure}

The comparison of the mentioned three scenarios is not completely straightforward, as we need to decide what quantities we are comparing.
In the current investigation we concentrate on measurements related to orthonormal bases and we assume that the noise is standard uniform noise, i.e., simple coin tossing noise. Within our framework this is equivalent to depolarizing noise. 
When it comes to the third scenario of one fixed and one free measurement, we pay particular attention to two alternative cases. In both cases the fixed measurement is related to the same basis. Its optimization however differs according to our knowledge or assumption on the free measurement. Namely, in the first case the free measurement is related to a completely arbitrary basis, while in the second case it is related to a basis that is unbiased with respect to the fixed one.
As one would expect, the noise trade-off is worse in the case with one fixed and one completely arbitrary basis. 
We prove that in this case it is exactly as bad as in the approximate cloning scenario, where both bases -- and not only the second one -- are not fixed in advance. Thus, there is no advantage in knowing the first basis and then arranging the measurement setup by enforcing this additional information. The result is surprising, as one would expect that postponing the choice of both bases is in fact more penalizing than postponing the choice of only one of them. No less surprisingly, we further prove that the case where the free basis is unbiased with respect to the fixed basis actually generates the same noise trade-off as the joint measurement scenario. Therefore, knowing that the two bases are mutually unbiased is enough to achieve the best possible noise trade-off even when the choice of one basis is postponed after the measurement of the other one. This means that the unbiasedness constraint allows to extract as much amount of information on the first basis as if both bases were fixed in advance, while it still enables to defer the choice of the second basis after the measurement of the first one. In conclusion, the two cases in which the second basis is free or it is unbiased with respect to the first basis are just the opposite extremes in the scenario where one of the bases is fixed in advance and the other basis is decided later.

An interesting additional viewpoint can be obtained by looking at overnoisy measurements, i.e., measurements in which the intensity of the noise is allowed to exceed unity. Indeed, mixing a quantum measurement with uniform noise formally makes sense even if the noise intensity is slightly greater than $1$. When the noise intensity takes its largest possible value, one obtains a measurement -- called the reverse version of the original measurement in \cite{FiHeLe17} -- which corresponds to discard the outcome of the original measurement and then output any other outcome chosen at random with the same probability. Remarkably, when overnoisy measurements are allowed, the scenario with one fixed and one arbitrary basis does actually improve the noise trade-off with respect to the approximate cloning scenario. Thus, one of the two results proved for the standard noise model fails for overnoisy measurements, namely, postponing the choice of both bases turns out to be more penalizing than postponing the choice of only one of them. The result holding for the standard noise model is then bolstered by the fact that it does not extend to overnoisy measurements, because as we already remarked its failure is just what one would expect in general. This should be compared with the other result proved for the standard noise model, which in fact is true also if we consider overnoisy measurements. Indeed, if the second basis is known to be unbiased with respect to the first one, we show that even in the overnoisy case postponing the choice of one basis still yields the same noise trade-off as in the joint measurement scenario. As an interesting byproduct of our analysis, we also find that, when the noise intensity exceeds unity, a striking difference emerges between joint measurement setups for qubit systems and systems having dimension higher than $2$.

We now provide an overview of the paper's structure. 
In Section \ref{sec:joint} we recall the concepts of approximate joint measurement and approximate cloning, describe their interplay and review some earlier known results. 
We call to mind that approximate joint measurements and approximate cloning constitute two different instances of compatibility for quantum devices, i.e., compatibility for quantum meters and compatibility for quantum channels, respectively. Section \ref{sec:sequential} contains the two main results of the paper, namely, it describes the two extreme cases in which postponing the choice of one measurement does not (Section \ref{sec:s1}) or does (Section \ref{sec:s2}) provide an improvement over the quantum cloning strategy. Overnoisy measurements are introduced in Section \ref{sec:exceptional}, where we show that not all the results of the previous section hold if the noise intensity exceeds unity. In Section \ref{sec:proof} we provide the proofs of the main results stated in Sections \ref{sec:sequential} and \ref{sec:exceptional}. Finally, Section \ref{sec:conc} contains our conclusions and some perspective on future works. We end this part by introducing the basic terminology and symbols that will be used throughout the paper.

\subsection*{Notations}
In this paper, we work in qudit systems, i.e., quantum systems whose associated Hilbert space is $\Cb^d$, where $d$ is the dimension of the system. We denote by $\lc{d}$ the $C^*$-algebra of all linear endomorphisms of $\Cb^d$, and we write $\id$ for the identity element and ${\rm tr}$ for the linear trace of $\lc{d}$. In a $d$-dimensional qudit system, a {\em state} is thus described by an element $\rhoo\in\lc{d}$ such that $\tr{\rhoo} = 1$ and $\rhoo\geq 0$, i.e., $\rhoo$ is positive semidefinite. 
For any finite set $X$, a \emph{meter} with outcomes in $X$ is a positive operator valued measure (POVM), i.e. a map $\Mo:X\to\lc{d}$ such that $\sum_x \Mo(x) = \id$ and $\Mo(x)\geq 0$ for all $x$. 
The probability of obtaining the outcome $x$ by measuring the meter $\Mo$ in the state $\rhoo$ is given by the Born rule $\tr{\rhoo\Mo(x)}$. 
We also consider transformations of one qudit system into another one, with possibly different dimensions $d_1$ and $d_2$ of the input and output systems. A transformation of this kind is described by a {\em channel}, i.e., a linear and completely positive (CP) map $\Phi:\lc{d_1}\to\lc{d_2}$ such that $\tr{\,\Phi(\rhoo)} = 1$ for all states $\rhoo\in\lc{d_1}$. 
Finally, by combining the notions of meter and channel, we obtain the definition of an {\em instrument}, that is, a map $\ii:X\times\lc{d_1}\to\lc{d_2}$ such that $\ii(x,\cdot)$ is linear and CP for all $x\in X$ and $\sum_x \ii(x,\cdot)$ is a channel. We remark that this is the Schr\"odinger picture for channels and instruments. 
The Heisenberg picture just amounts to switch from $\Phi$ and $\ii$ to the adjoint channel $\Phi^\dag:\lc{d_2}\to\lc{d_1}$ and the adjoint instrument $\ii^\dag:X\times\lc{d_2}\to\lc{d_1}$, which are defined by $\tr{\rhoo\Phi^\dag(A)} = \tr{\Phi(\rhoo)A}$ and $\tr{\rhoo\ii^\dag(x,A)} = \tr{\ii(x,\rhoo)A}$ for all $\rhoo\in\lc{d_1}$, $A\in\lc{d_2}$ and $x\in X$.
\section{Approximate joint meters and cloners}\label{sec:joint}

Suppose $\Mo$ and $\No$ are two meters on a qudit system, with the respective outcome sets $X$ and $Y$. Then, $\Mo$ and $\No$ can be measured in a single experimental run exactly when there exists a third meter $\Go$ that has the Cartesian product $X\times Y$ as its outcome set and satisfies the relations
\begin{equation} 
\Mo(x) = \sum_y\Go(x,y)\,,\qquad \No(y) = \sum_x\Go(x,y)
\end{equation}
for all $x\in X$ and $y\in Y$. If this is the case, we say that the meters $\Mo$ and $\No$ are \emph{compatible}, and $\Go$ is a \emph{joint meter} for them. 
On the other hand, if $\Mo$ and $\No$ are incompatible (i.e., not compatible), we can add noise to both and observe that at a certain threshold value the so obtained noisy meters become compatible. The threshold value depends on $\Mo$ and $\No$ and also on the type of noise \cite{DeFaKa19}. 
In our following investigation we focus on the simplest type of noise, that is, the uniform noise.

To provide a concrete example, we fix the target meters $\Qo$ and $\Po$ associated with two different orthonormal bases $\{\phii_0,\ldots,\phii_{d-1}\}$ and $\{\psi_0,\ldots,\psi_{d-1}\}$, i.e.,
\begin{equation}\label{eq:QP}
\Qo(x)=\kb{\varphi_x}{\varphi_x} \, , \qquad \Po(y)=\kb{\psi_y}{\psi_y}
\end{equation}
for all $x,y$ belonging to the outcome sets $X=Y=\{0,1,\ldots,d-1\}$. 
Since $\Qo$ and $\Po$ are {\em sharp meters} (i.e., consist of rank-$1$ projections) and they do not commute, they are incompatible. 
However, we can add uniform noise to both and thus obtain the {\em noisy versions}
\begin{equation}\label{eq:QsPt}
\Qo_s(x) = s\,\kb{\varphi_x}{\varphi_x} + (1-s)\tfrac{1}{d}\id \,,\qquad \Po_t(y) = t\,\kb{\psi_y}{\psi_y} + (1-t)\tfrac{1}{d}\id\,,
\end{equation}
where the amount of noise is characterized by parameters $s,t\in [0,1]$. For sufficiently small values of $s$ and $t$, the meters $\Qo_s$ and $\Po_t$ are indeed compatible, and any joint meter for them is therefore an \emph{approximate joint meter} of $\Qo$ and $\Po$. 
The obvious question is how large $s$ an $t$ can be for a joint meter to exist. 
To answer this question, the essential point is that the choice of the two bases is made at once in the present scenario, and the approximate joint meter can then be optimized depending on both bases.

In our previous example the noise parameters $s$ and $t$ can be different, thus accounting for the trade-off between the accuracies in the approximation of $\Qo$ and the approximation of $\Po$ that are achievable by means of any approximate joint meter. This trade-off is always weaker than the relation $s+t \leq 1$, and thus the triangle with vertices $(0,0)$, $(1,0)$ and $(0,1)$ is always contained in the \emph{compatibility region} of the pair $(\Qo,\Po)$, i.e., the subset
\begin{equation*}
\CR{(\Qo,\Po)} = \{(s,t) \in [0,1]\times [0,1] \mid \Qo_s \text{ and } \Po_t \text{ are compatible}\}
\end{equation*}
of the unit square $[0,1]\times [0,1]$ (see \cite{BuHeScSt13}).

For the most part of this paper, we consider the case in which the sharp meters $\Qo$ and $\Po$ are \emph{mutually unbiased}, that is, the equality $\tr{\Qo(x)\Po(y)} = \frac{1}{d}$ is satisfied for all $x,y$. This amounts to the requirement that the two orthonormal bases $\{\phii_0,\ldots,\phii_{d-1}\}$ and $\{\psi_0,\ldots,\psi_{d-1}\}$ of \eqref{eq:QP} are mutually unbiased in the usual sense \cite{WoFi89}. Under the unbiasedness assumption, the compatibility region of the pair $(\Qo,\Po)$ was determined in \cite{CaHeTo12,CaHeTo19,CaCaTo19} (see \cite[Eq.~(5) and Fig.~1]{CaCaTo19}). Theorem \ref{thm:QP+II} below and the subsequent discussion summarize this earlier result.

An obvious method to measure two incompatible meters on the same qudit system is performing it by means of cloning. Actually, if perfect cloning of quantum states would be possible, then clearly all meters would be compatible, as we could make copies and perform measurements on them \cite{QI01Werner}. 
Even if perfect cloning is impossible, we may try to use an imperfect but realizable cloning device as a way to jointly measure approximations of incompatible meters \cite{HeScToZi14}.
Namely, we use a channel $\Gamma:\lc{d}\to\elle{\Cb^d\otimes\Cb^d}$ to produce two approximate clones $\rhoo_1 = \ptr{\bar{1}}{\Gamma(\rhoo)$ and $\rhoo_2 = \ptr{\bar{2}}{\Gamma(\rhoo)}$ of the unknown initial state $\rhoo$, where $\bar{1}=2$, $\bar{2}=1$ and ${\rm tr}_{i}$ is the partial trace over the $i$-th subsystem of $\Cb^d\otimes\Cb^d$.} 
Then, we measure two meters $\Mo$ and $\No$ separately on these approximate clones. The resulting meter $\Go^\Gamma_{\Mo,\No}$ is thus given by
\begin{equation} 
\Go^\Gamma_{\Mo,\No}(x,y) = \Gamma^\dag(\Mo(x)\otimes\No(y)) \, .
\end{equation}
The essential point is that $\Go^\Gamma_{\Mo,\No}$ is not a joint meter for $\Mo$ and $\No$, but for two meters $\Mo'$ and $\No'$ that are close to $\Mo$ and $\No$, in some sense.

If we apply the procedure described above to the sharp meters $\Qo$ and $\Po$ defined in \eqref{eq:QP} and we further make some hypotheses on the cloning device $\Gamma$, the resulting meter $\Go^\Gamma_{\Qo,\Po}$ turns out to be an approximate joint meter of $\Qo$ and $\Po$. 
It is then easy to compare the original pair $(\Qo,\Po)$ with the pair $(\Qo',\Po')=(\Qo_s,\Po_t)$ obtained by means of approximate cloning, as it is simply a matter of evaluating the noise parameters $s$ and $t$. In particular, we can optimize $\Gamma$ over all possible choices of the sharp meters $\Qo$ and $\Po$, thus obtaining a universal device to make two arbitrary sharp meters noisy enough to be compatible.

If we require that $\Go^\Gamma_{\Qo,\Po}$ is an approximate joint meter of $\Qo$ and $\Po$ for any possible choice of the sharp meters $\Qo$ and $\Po$, then the approximate clones must have the form $\rhoo_1 = s\rhoo + (1-s)\tfrac{1}{d}\id$ and $\rhoo_2 = t\rhoo + (1-t)\tfrac{1}{d}\id$ for $s,t\in [0,1]$ (see Appendix \ref{app:noise} for the details).
Equivalently, the approximate cloning device $\Gamma$ needs to satisfy the relations
\begin{equation}\label{eq:joint_Gamma}
\ptr{\bar{1}}{\circ\Gamma} = s I + (1-s)\tfrac{1}{d}\id{\rm tr}\,,\qquad \ptr{\bar{2}}{\circ\Gamma} = t I + (1-t)\tfrac{1}{d}\id{\rm tr}\,,
\end{equation}
in which $I:\lc{d}\to\lc{d}$ is the identity channel. 
The right-hand sides of the last two expressions are instances of the (partially) {\em depolarizing channel}
\begin{equation}\label{eq:Ir}
I_r = r I + (1-r)\tfrac{1}{d}\id{\rm tr}
\end{equation}
for different values of the noise parameter $r\in [0,1]$. A channel $\Gamma$ which satisfies \eqref{eq:joint_Gamma} is a {\em joint channel} for $I_s$ and $I_t$ and its existence means that $I_s$ and $I_t$ are {\em compatible channels} in the sense that they can be implemented with a single device \cite{HeMi17}. 
The compatibility region of the pair $(I,I)$ is then the set
\begin{equation*}
\CR{(I,I)} = \{(s,t) \in [0,1]\times [0,1] \mid I_s \text{ and } I_t \text{ are compatible}\}\,.
\end{equation*}
We observe that the inclusion $\CR{(I,I)}\subseteq \CR{(\Qo,\Po)}$ holds, since any joint channel $\Gamma$ for $I_s$ and $I_t$ yields a joint meter $\Go^\Gamma_{\Qo,\Po}$ for $\Qo_s$ and $\Po_t$. 

The compatibility region of the pair $(I,I)$ was derived in \cite{Hashagen17,Haapasalo19} (see \cite[Eq.~(24) and Fig.~B.1]{Hashagen17} and \cite[Eq.~(5.12) and Fig.~1]{Haapasalo19}). The following theorem reports this result, along with the characterization of $\CR{(\Qo,\Po)}$ for two mutually unbiased sharp meters $\Qo$ and $\Po$.

\begin{theorem}\label{thm:QP+II}
Suppose $(s,t)\in [0,1]\times [0,1]$.
\begin{enumerate}[(a)]
\item For any two mutually unbiased sharp meters $\Qo$ and $\Po$, the respective noisy versions $\Qo_s$ and $\Po_t$ are compatible if and only if
\begin{equation}\label{eq:QP_comp}
s+t-1\leq \frac{2}{\sqrt{d}}\sqrt{(1-s)(1-t)}\, .
\end{equation}
\item The depolarizing channels $I_s$ and $I_t$ are compatible if and only if
\begin{equation}\label{eq:II_comp}
s+t-1\leq \frac{2}{d}\sqrt{(1-s)(1-t)}\, .
\end{equation}
\end{enumerate}
\end{theorem}

As a consequence of Theorem \ref{thm:QP+II}, the compatibility regions of any pair of mutually unbiased sharp meters $(\Qo,\Po)$ and of the pair of identity channels $(I,I)$ are symmetric along the $s=t$ line. 
We also observe that $\CR{(I,I)}$ in dimension $d$ is the same as $\CR{(\Qo,\Po)}$ in dimension $d^2$.

The most important elements of the compatibility region are its extreme points, as those correspond to the optimal ways of making an approximate joint implementation of the two incompatible devices. Apart from the trivial point $(0,0)$, all extreme points $(s,t)$ are located where the equality is attained in \eqref{eq:QP_comp} (respectively, in \eqref{eq:II_comp}), i.e., on the arc segment obtained by intersecting the half-plane $s+t\geq 1$ and the ellipse
\begin{equation}\label{eq:poly}
d^k (s^2 + t^2) + 2(d^k -2)(1-s)(1-t) = d^k
\end{equation}
with $k=1$ (resp., with $k=2$). See Figure \ref{fig:comp}.

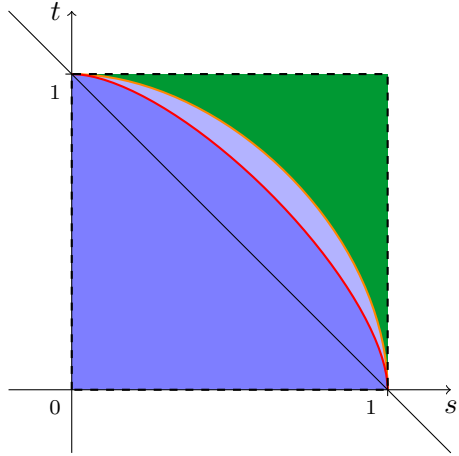
\begin{figure}[h!]\def\d{3}
\centering
\begin{tikzpicture}[scale=4.2]

\fill[color=darkgreen]
(0,0) -- (0,1) -- (1,1) -- (1,0) -- (0,0) ;

\foreach \k/\area/\curve in {\d/lightblue/darkorange,{\d^2}/darkblue/red}{

\def\th0{(pi-ratan(sqrt(\k-1)))}

\fill[domain={-ratan(sqrt(\k-1))}:{ratan(sqrt(\k-1))}, samples=40, color=\area]
(0,0) -- plot ({(\k * rcos(\x + \th0) + 2 - \k)/(2 * (1 - \k))}, {(\k * rcos(\x - \th0) + 2 - \k)/(2*(1 - \k))}) -- (0,0) ;

\draw[thick, domain={-ratan(sqrt(\k-1))}:{ratan(sqrt(\k-1))}, samples=40, color=\curve]
plot ({(\k * rcos(\x + \th0) + 2 - \k)/(2 * (1 - \k))}, {(\k * rcos(\x - \th0) + 2 - \k)/(2*(1 - \k))}) ;
}

\draw[thick, dashed]
(0,0) -- (0,1) -- (1,1) -- (1,0) -- (0,0) ;

\draw[->]
(-0.2,0) -- (0,0)node[anchor=north east]{\tiny $0$} -- (1.2,0)node[anchor=north]{\small $s$} ;

\draw[->]
(0,-0.2) -- (0,1.2)node[anchor=east]{\small $t$} ;

\draw (1,0)node[anchor=north east]{\tiny $1$};
\draw[thin,-] (1,-0.02)--(1,0.02);

\draw (0,1)node[anchor=north east]{\tiny $1$};
\draw[thin,-] (-0.02,1)--(0.02,1);

\draw
(1.2,-0.2) -- (-0.2,1.2) ;

\end{tikzpicture}
\caption{The set of points $(s,t)\in [0,1]\times [0,1]$ (green square) and the compatibility regions $\CR{(\Qo,\Po)}$ and $\CR{(I,I)}$ given by Theorem \ref{thm:QP+II} (blue areas) in dimension $d=3$. Lighter blue denotes the set difference $\CR{(\Qo,\Po)}\setminus\CR{(I,I)}$. Points of the orange (respectively, red) curve attain the equalities in \eqref{eq:QP_comp} (resp., \eqref{eq:II_comp}).\label{fig:comp}}
\end{figure}

We now describe the joint meters and channels which correspond to the extreme points of the compatibility regions written in Theorem \ref{thm:QP+II}. 
They will be our benchmarks in comparison with the sequential measurement scenario described in the next section.
To do it, for $k=1,2$ and all $r\in[0,1]$, we introduce the quantities $a_k(r),b(r)\in [0,1]$ defined as
\begin{equation}\label{eq:ab_bis}
a_k(r) = \frac{1}{\sqrt{d^k}}\left[\sqrt{r(d^k-1)+1} - \sqrt{1-r}\right]\,,\qquad b(r) = \sqrt{1-r}\,.
\end{equation}
These quantities satisfy the identity
\begin{equation}\label{eq:ab_ellipse}
a_k(r)^2 + b(r)^2 + \frac{2}{\sqrt{d^k}} a_k(r)b(r) = 1
\end{equation}
and, for $s,t\in [0,1]$, they fulfill the equivalences
\begin{equation}\label{eq:a=b_equiv}
\begin{aligned}
& a_k(s) = b(t) \quad\Leftrightarrow\quad a_k(t) = b(s) \\
& \qquad\Leftrightarrow\quad s+t-1 = \frac{2}{\sqrt{d^k}}\sqrt{(1-s)(1-t)}\,.
\end{aligned}
\end{equation}
For $s\in [0,1]$, we define the meter $\Go^{(s)}$ with the outcome set $\{0,1,\ldots,d-1\}\times\{0,1,\ldots,d-1\}$ as done in \cite[Eq.~(11)]{CaCaTo19}, i.e.,
\begin{equation}\label{eq:optGo}
\begin{aligned}
\Go^{(s)}(x,y) = {} & \Qo_s(x)^{\frac{1}{2}} \Po(y) \Qo_s(x)^{\frac{1}{2}} \\
= {} & a_1(s)^2\tr{(\Qo(x)\Po(y))} \Qo(x) + \frac{1}{d} b(s)^2 \Po(y) \\
& + \frac{1}{\sqrt{d}} a_1(s) b(s) \big(\Qo(x)\Po(y) + \Po(y)\Qo(x)\big) \,,
\end{aligned}
\end{equation}
where $\Qo$ and $\Po$ are the sharp meters \eqref{eq:QP} and
\begin{equation}\label{eq:Q1/2}
\Qo_s(x)^{\frac{1}{2}} = a_1(s)\Qo(x) + \frac{1}{\sqrt{d}} b(s) \id\,.
\end{equation}
If $\Qo$ and $\Po$ are mutually unbiased, then \eqref{eq:ab_ellipse} and \eqref{eq:optGo} imply
\begin{equation}\label{eq:optGo_margins}
\sum_y\Go^{(s)}(x,y) = \Qo_{1-b(s)^2}(x)\,,\qquad \sum_x\Go^{(s)}(x,y) = \Po_{1-a_1(s)^2}(y)\,.
\end{equation}
Therefore, if the equality is attained in \eqref{eq:QP_comp}, the definition of $b$ and the equivalences \eqref{eq:a=b_equiv} entail that $\Go^{(s)}$ is a joint meter for $\Qo_s$ and $\Po_t$.

The optimal quantum cloning device of \cite[Eq.~(21)]{Hashagen17} is given by
\begin{equation}\label{eq:optGamma}
\Gamma^{(s)}(\rho) = \big(a_2(s)\id\otimes\id + b(s)F\big)\left(\rhoo\otimes\tfrac{1}{d}\id\right)\big(a_2(s)\id\otimes\id + b(s)F\big)
\end{equation}
for all states $\rhoo$, where $s\in [0,1]$ is a parameter and $F\in\elle{\Cb^d\otimes\Cb^d}$ is the flip operator satisfying
\begin{equation} 
F(\xi\otimes\zeta) = \zeta\otimes\xi
\end{equation}
for all $\xi,\zeta\in\Cb^d$. The following analogue of \eqref{eq:optGo_margins} holds for the channel $\Gamma^{(s)}$
\begin{equation}\label{eq:optGamma_margins}
\ptr{\bar{1}}{\circ\Gamma^{(s)}} = I_{1-b(s)^2}\,,\qquad \ptr{\bar{2}}{\circ\Gamma^{(s)}} = I_{1-a_2(s)^2}\,.
\end{equation}
Therefore, if the equality is attained in \eqref{eq:II_comp}, then $\Gamma^{(s)}$ is a joint channel for $I_s$ and $I_t$ again as a consequence of the definition of $b$ and the equivalences \eqref{eq:a=b_equiv}.

To focus on the case of symmetric noise, we remark that the equality $a_k(s) = b(s)$ holds if and only if $s=(\sqrt{d^k}+2)/[2(\sqrt{d^k}+1)]$. By using this fact in \eqref{eq:optGo_margins}, we recover that $\Qo_s$ and $\Po_s$ are compatible if and only if $s=(\sqrt{d}+2)/[2(\sqrt{d}+1)]$ \cite{DeSkFrBr19}. By using the same fact in \eqref{eq:optGamma_margins}, we recover the optimal symmetric universal quantum cloning device
\begin{equation*}
\Gamma^{\left(\frac{d+2}{2(d+1)}\right)}(\varrho) = \frac{2}{d+1} S (\varrho \otimes \id) S \,, 
\end{equation*}
where $S = \frac{1}{2}(\id\otimes\id + F)$ is the projection onto the symmetric subspace of $\Cb^d\otimes\Cb^d$ \cite{KeWe99}.

\section{Sequential measurements}\label{sec:sequential}

The third scenario between the earlier two is the one where first of the target meters, say $\Mo$, is fixed from the beginning, but the second target meter $\No$ is decided only later. 
In this scenario we first perform a measurement of an approximation $\Mo'$ of $\Mo$, then in the resulting state we subsequently measure $\No$. 
A measurement of $\Mo'$ necessarily disturbs the state of the qudit system in some way, unless $\Mo'$ is totally noninformative. 
However, measuring the approximation $\Mo'$ can be made less destructive than a direct measurement of $\Mo$ \cite{HeMi13}.

The sequential measurement scheme gives an advantage over fixed joint measurements of two meters, because the second meter $\No$ may as well be decided only after one has obtained information on the system based on the outcome of $\Mo'$. 
However, it is not as universal as the approximate cloning setup, since one of the meters is now fixed and the first measurement can be tailored for that specific meter.

The statistics of a meter $\Mo'$ as well as the statistics of any meter $\No$ measured on the qudit system after $\Mo'$ are both described by an instrument $\ii:X\times\lc{d}\to\lc{d}$. Namely, for any prior state $\rhoo$ of the system, the outcome $x$ of $\Mo'$ is obtained with probability $\tr{\ii(x,\rhoo)}$, and the posterior state conditioned to $x$ is $\ii(x,\rhoo)/\tr{\ii(x,\rhoo)}$. Thus, if after $\Mo'$ we subsequently perform a measurement of $\No$, the probability of obtaining the outcomes $x$ and $y$ in succession is $\tr{\ii(x,\rhoo)\No(y)} = \tr{\rhoo\Go^\ii_\No(x,y)}$, where $\Go^\ii_\No(x,y) = \ii^\dag(x,\No(y))$. The overall resulting meter $\Go^\ii_\No$ is then a joint meter for $\Mo'$ and a disturbed version $\No'$ of $\No$, i.e., $\No'(y) = \sum_x\ii^\dag(x,\No(y))$.

\subsection{The second choice is arbitrary}\label{sec:s1}

If in the sequential measurement scheme described above the target meters are the sharp meters $\Qo$ and $\Po$ of \eqref{eq:QP}, we can require that the resulting meter $\Go^\ii_\Po$ is an approximate joint meter of $\Qo$ and $\Po$. Indeed, as it happens with approximate cloning, it is then easy to compare the original pair $(\Qo,\Po)$ with its approximation $(\Qo',\Po') = (\Qo_s,\Po_t)$. 
The essential difference from the cloning scenario is that in the present setting $\Qo$ is fixed, and only $\Po$ is allowed to vary among all sharp meters. Requiring that $\Go^\ii_\Po$ is an approximate joint meter of $\Qo$ and $\Po$ for any possible choice of the second meter $\Po$ then amounts to impose that the instrument $\ii$ satisfies the two conditions
\begin{equation}\label{eq:joint_J}
\ii^\dag(\cdot,\id) = \Qo_s\,,\qquad \sum_x \ii(x,\cdot) = I_t
\end{equation}
for some $s,t\in [0,1]$ (see Appendix \ref{app:noise}).
An instrument $\ii$ which satisfies \eqref{eq:joint_J} is a {\em joint instrument} for $\Qo_s$ and $I_t$, and the meter $\Qo_s$ and the channel $I_t$ are {\em compatible} whenever they admit a joint instrument \cite{HeMi13,HeMiRe14}.

As we did for the pairs $(\Qo,\Po)$ and $(I,I)$, we introduce the compatibility region of the pair $(\Qo,I)$, i.e., the set
\begin{equation*}
\CR{(\Qo,I)} = \{(s,t) \in [0,1]\times [0,1] \mid \Qo_s \text{ and } I_t \text{ are compatible}\}\,.
\end{equation*}
We then observe that the relation $\CR{(I,I)}\subseteq\CR{(\Qo,\Po)}$ can be refined to the chain of inclusions
\begin{equation}\label{eq:chain}
\CR{(I,I)}\subseteq\CR{(\Qo,I)}\subseteq\CR{(\Qo,\Po)}\,.
\end{equation}
Indeed, the second inclusion is implied by the sequential measurement scheme, since the resulting meter $\Go^\ii_\Po$ is a joint meter for $\Qo_s$ and $\Po_t$ whenever $\ii$ is a joint instrument for $\Qo_s$ and $I_t$. On the other hand, to prove the first inclusion, we consider a further measurement scheme in which we produce two approximate clones of the initial state $\rhoo$ and we subsequently perform a measurement of $\Qo$ only on the first clone. Thus, there are two outputs, i.e., the outcome $x$ of $\Qo$ and the second clone conditioned to $x$. The scheme is then described by an instrument $\ii^\Gamma_\Qo$, where $\Gamma$ denotes the cloning device and $\ii^\Gamma_\Qo(x,\rhoo) = \ptr{\bar{2}}{\left[(\Qo(x)\otimes\id)\Gamma(\rhoo)\right]}$. In particular, if $\Gamma$ is a joint channel for $I_s$ and $I_t$, the resulting instrument $\ii^\Gamma_\Qo$ is a joint instrument for $\Qo_s$ and $I_t$, hence $\CR{(I,I)}\subseteq\CR{(\Qo,I)}$ as claimed.

Cloning $\rhoo$ and then performing a measurement of $\Qo$ on the first clone is only one out of the many ways to construct a joint instrument for $\Qo_s$ and $I_t$. For this reason, one would expect the first inclusion of \eqref{eq:chain} to be strict. Surprisingly, the next theorem shows that this is not the case. The
proof relies on some basic applications of finite Fourier transform and Weyl-Heisenberg group, and we defer it to Section \ref{sec:proof}.

\begin{theorem}\label{thm:main}
Suppose $(s,t)\in [0,1]\times [0,1]$. For any sharp meter $\Qo$, the noisy version $\Qo_s$ and the depolarizing channel $I_t$ are compatible if and only if
\begin{equation}\label{eq:QI_comp}
s+t-1\leq \frac{2}{d}\sqrt{(1-s)(1-t)}\, .
\end{equation}
\end{theorem}

The comparison of Theorems \ref{thm:QP+II} and \ref{thm:main} yields a further refinement of the chain of inclusions \eqref{eq:chain}, namely,
\begin{equation} 
\CR{(I, I)} = \CR{(\Qo, I)} \subset \CR{(\Qo,\Po)} \, .
\end{equation}
In particular, the equality $\CR{(I, I)} = \CR{(\Qo, I)}$ leads to the first main result of the paper.

\begin{result}\label{res:1}
Suppose the target pairs $(\Qo,\Po)$ are chosen among all possible pairs of sharp meters. Then, the two strategies consisting in
\begin{enumerate}[(i)]
\item optimally cloning the qudit state for the purpose of jointly measuring any pair $(\Qo,\Po)$, and
\item performing a sequential measurement that is optimized for those pairs $(\Qo,\Po)$ in which $\Qo$ is fixed
\end{enumerate}
actually gives the same trade-off between the achievable approximations of $\Qo$ and $\Po$. When $\Po$ is arbitrary, there is hence no advantage in knowing $\Qo$ in advance, and the cloning strategy is still optimal even in the scenario where the choice of $\Po$ is postponed until an approximate measurement of $\Qo$ has been performed.
\end{result}

In order to describe the joint instruments which correspond to the extreme points of the set $\CR{(\Qo,I)}$, we implement the measurement scheme in which we first apply the optimal quantum cloning device $\Gamma^{(s)}$ of \eqref{eq:optGamma} and we subsequently perform a measurement of $\Qo$ on the first approximate clone. The resulting instrument $\ii^{(s)} \equiv \ii^{\Gamma^{(s)}}_\Qo$ is then
\begin{equation}\label{eq:optJ}
\begin{aligned}
& \ii^{(s)}(x,\rhoo) = \ptr{\bar{2}}{\left[(\Qo(x)\otimes\id)\Gamma^{(s)}(\rhoo)\right]} \\
& \quad = \frac{1}{d} \big\{ a_2(s)^2 \tr{(\rhoo\Qo(x))} \id + a_2(s)b(s) (\rhoo\Qo(x) + \Qo(x)\rhoo) + b(s)^2 \rhoo \big\} \,.
\end{aligned}
\end{equation}
It satisfies the relations
\begin{equation} 
\ii^{(s)\,\dag}(\cdot,\id) = \Qo_{1-b(s)^2}\,,\qquad \sum_x \ii^{(s)}(x,\cdot) = I_{1-a_2(s)^2}\,.
\end{equation}
Therefore, if the equality is attained in \eqref{eq:QI_comp}, then $\ii^{(s)}$ is a joint instrument for $\Qo_s$ and $I_t$ by an argument similar to those following \eqref{eq:optGo_margins} and \eqref{eq:optGamma_margins}.


\subsection{The second choice is unbiased with respect to the first choice}\label{sec:s2}

In the sequential measurement scheme considered before, the target meters $\Qo$ and $\Po$ were arbitrary sharp meters. However, if the second meter $\Po$ is not completely arbitrary, but it is rather chosen among a restricted set of possibilities, we may be able to optimize the measurement of $\Qo_s$ by enforcing any available additional constraint on $\Po$. 
Obviously, the condition $\sum_x \ii(x,\cdot) = I_t$ does no longer apply in this case, as the choice of $\Po$ privileges some orthonormal bases. For example, if we know in advance that $\Po$ is unbiased with respect to $\Qo$, we can employ the following {\em L\"uders instrument}
\begin{equation}\label{eq:Luders}
\ii_L^{(s)}(x,\rhoo) = \Qo_s(x)^{\frac{1}{2}} \rhoo \Qo_s(x)^{\frac{1}{2}}
\end{equation}
to perform a measurement of $\Qo_s$. In this way, if we subsequently measure the unbiased meter $\Po$, the resulting joint meter is 
\begin{equation}\label{eq:Luders+P}
\ii_L^{(s)\,\dag}(x,\Po(y)) = \Go^{(s)}(x,y)
\end{equation}
for all $x,y$, i.e., the optimal approximate joint meter \eqref{eq:optGo} of $\Qo$ and $\Po$. 
Therefore, when the unbiasedness condition is satisfied, the optimal approximate joint meter of $\Qo$ and $\Po$ can be implemented in a sequential measurement scheme, which is then the best possible sequential measurement scheme for the target meters $\Qo$ and $\Po$ \cite{CaCaTo19}. Due to the strict inclusion $\CR{(\Qo,I)}\subset\CR{(\Qo,\Po)}$, this scheme allows a better approximation of the pair $(\Qo,\Po)$ compared with the case in which $\Po$ is arbitrary, even in the scenario where the choice of $\Po$ is postponed after the measurement of $\Qo_s$. We thus achieve the second main result of the paper.
\begin{result}\label{res:2}
Suppose the target pairs $(\Qo,\Po)$ are chosen among all pairs of mutually unbiased sharp meters. Then, the two strategies consisting in
\begin{enumerate}[(i)]
\item measuring an optimal approximate joint meter of a specific pair $(\Qo,\Po)$, and
\item performing a sequential measurement that is optimized for those pairs $(\Qo,\Po)$ in which only $\Qo$ is fixed
\end{enumerate}
yield the same trade-off between the achievable approximations of $\Qo$ and $\Po$. When $\Po$ is unbiased with respect to $\Qo$, there is hence an advantage in knowing $\Qo$ in advance, and the cloning strategy is no longer optimal in the scenario where the choice of $\Po$ is postponed until an approximate measurement of $\Qo$ has been performed.
\end{result}

Remarkably, if $(s,t)$ attains the equality in \eqref{eq:QI_comp}, the L\"uders instrument \eqref{eq:Luders} is a joint instrument for the meter $\Qo_s$ and the channel $I'_t$ given by
\begin{equation} 
I'_t(\rhoo) = t\rhoo + (1-t) \sum_x \tr{\rhoo\Qo(x)} \Qo(x)\,.
\end{equation}
This channel differs from a depolarizing channel $I_t$ in that the depolarizing noise $\tfrac{1}{d}\id {\rm tr}$ of \eqref{eq:Ir} is replaced by the {\em measure and prepare channel} $\Phi_\Qo$ defined as
\begin{equation} 
\Phi_\Qo(\rhoo) = \sum_x \tr{\rhoo\Qo(x)} \Qo(x) \,.
\end{equation}
If the second meter $\Po$ is mutually unbiased with respect to $\Qo$, then the measure-and-prepare channel $\Phi_\Qo$ has the same effect on it as the depolaraizing channel.

\section{Overnoisy devices and extended compatibility regions}\label{sec:exceptional}

There are intriguing differences when we go out of the previously discussed noise model and we allow the noise parameters $s,t$ to take also negative values. 
Indeed, for any two sharp meters $\Qo$ and $\Po$ as in \eqref{eq:QP}, the noisy versions $\Qo_s$ and $\Po_t$ defined in \eqref{eq:QsPt} not only make sense for $s,t\in [0,1]$, but also for sufficiently small negative values of $s$ and $t$. 
In fact, $\Qo_s$ and $\Po_t$ are valid meters also for $s,t$ in the half-open interval $[m_1,0)$, where $m_1 = -1/(d-1)$ is the minimal value of $s$ and $t$ such that $\Qo_s(x)\geq 0$ and $\Po_t(y)\geq 0$ for all $x,y$. In the extreme case $s=t=m_1$ we have
$$
\Qo_{m_1}(x) = \frac{1}{d-1} \left( \id - \Qo(x) \right)\,,\qquad \Po_{m_1}(y) = \frac{1}{d-1} \left( \id - \Po(y) \right)\,.
$$
These are called the \emph{reverse versions} of $\Qo$ and $\Po$, respectively \cite{FiHeLe17}. For $s,t$ in the open interval $(m_1,1)$, the pair $(\Qo_s,\Po_t)$ is then a convex mixture of the original pair $(\Qo,\Po)$ and the reverse pair $(\Qo_{m_1},\Po_{m_1})$.

Similarly as with meters, the depolarizing channel $I_r$ defined in \eqref{eq:Ir} not only makes sense for $r\in [0,1]$, but it is still a valid channel also for $r\in [m_2,0)$, where $m_2 = -1/(d^2-1)$. Indeed, $I_r$ is trace-preserving for all real $r$'s, and $r=m_2$ is the minimal value such that $I_r$ is a CP map \cite{King03}. The channel $I_{m_2}$ is known to be the {\em structural physical approximation} of the {\em reduction map} of $\lc{d}$, i.e., it is the closest channel to the positive trace preserving but not CP map $I_{m_1}$ that sends any pure state $\kb{\xi}{\xi}$ into the orthogonal state $[1/(d-1)](\id-\kb{\xi}{\xi})$ (see \cite{Bae17} and references therein). In dimension $d=2$ the channel $I_{m_2}$ is therefore the best realizable approximation of the (hypothetical) universal quantum NOT gate, as described in \cite{BuHiWe99}.

As we are mainly interested in pairs of quantum devices, we say that the product sets $[m_1,1]\times[m_1,1]$ and $[m_2,1]\times[m_2,1]$ constitute the squares of {\em admissible points} for the pair $(\Qo,\Po)$ and the pair $(I,I)$, respectively (see Figure \ref{fig:squares_1}).

\subsection{Compatibility regions and Result \ref{res:1} for negative noise parameters}\label{sec:ext1}

Within all admissible points $(s,t)$, we can search for the possibly negative values of $s,t$ such that either the meters $\Qo_s$ and $\Po_t$ or the channels $I_s$ and $I_t$ are compatible. This leads us to introduce the {\em extended compatibility regions}
\begin{align*}
\ECR{(\Qo,\Po)} & = \left\{(s,t) \in [m_1,1]\times [m_1,1] \mid \Qo_s \text{ and } \Po_t \text{ are compatible}\right\}\,,\\
\ECR{(I,I)} & = \left\{(s,t) \in [m_2,1]\times [m_2,1] \mid I_s \text{ and } I_t \text{ are compatible}\right\}\,.
\end{align*}
If instead we consider the mixed type devices $\Qo$ and $I$, the adimissible points for the pair $(\Qo,I)$ form the rectangle $[m_1,1]\times[m_2,1]$ (see Figure \ref{fig:squares_2}). Inside this rectangle the extended compatibility region of $(\Qo,I)$ is defined in a similar way as before, i.e.,
\begin{equation*}
\ECR{(\Qo,I)} = \left\{(s,t) \in [m_1,1]\times [m_2,1] \mid \Qo_s \text{ and } I_t \text{ are compatible}\right\}\,.
\end{equation*}

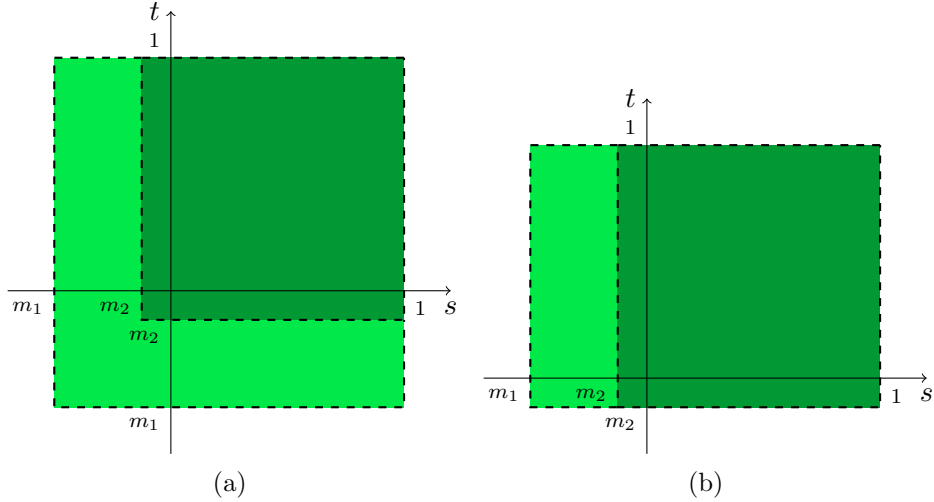
\begin{figure}[h!]\def\d{3}
\centering
\subfigure[\label{fig:squares_1}]{\begin{tikzpicture}[scale=3.1]

\foreach \k/\area in {1/lightgreen,2/darkgreen}{

\fill[color=\area]
({-1/(\d^\k-1)},{-1/(\d^\k-1)}) -- ({-1/(\d^\k-1)},1) -- (1,1) -- (1,{-1/(\d^\k-1)}) -- ({-1/(\d^\k-1)},{-1/(\d^\k-1)}) ;

\draw[thick, dashed]
\ifnum\k=1 ({-1/(\d^2-1)},1) \else (1,1) \fi -- ({-1/(\d^\k-1)},1) -- ({-1/(\d^\k-1)},{-1/(\d^\k-1)}) -- (1,{-1/(\d^\k-1)}) -- \ifnum\k=1 (1,{-1/(\d^2-1)}) \else (1,1) \fi ;

\draw ({-1/(\d^\k-1)},0)node[anchor=north east]{\tiny $m_\k$};
\draw (0,{-1/(\d^\k-1)})node[anchor=north east]{\tiny $m_\k$};

}

\draw[->]
({-0.2-1/(\d-1)},0) -- (1.2,0)node[anchor=north]{\small $s$} ;
\draw[->]
(0,{-0.2-1/(\d-1)}) -- (0,1.2)node[anchor=east]{\small $t$} ;

\draw (1,0)node[anchor=north west]{\tiny $1$};
\draw (0,1)node[anchor=south east]{\tiny $1$};

\end{tikzpicture}}
\subfigure[\label{fig:squares_2}]{\begin{tikzpicture}[scale=3.1]

\foreach \k/\area in {1/lightgreen,2/darkgreen}{

\fill[color=\area]
({-1/(\d^\k-1)},{-1/(\d^2-1)}) -- ({-1/(\d^\k-1)},1) -- (1,1) -- (1,{-1/(\d^2-1)}) -- ({-1/(\d^\k-1)},{-1/(\d^2-1)}) ;

\draw[thick, dashed]
\ifnum\k=1 ({-1/(\d^2-1)},1) \else (1,1) \fi -- ({-1/(\d^\k-1)},1) -- ({-1/(\d^\k-1)},{-1/(\d^2-1)}) -- \ifnum\k=1 ({-1/(\d^2-1)},{-1/(\d^2-1)}) \else (1,{-1/(\d^2-1)}) -- (1,1) \fi ;

\draw ({-1/(\d^\k-1)},0)node[anchor=north east]{\tiny $m_\k$};
\ifnum\k=2 \draw (0,{-1/(\d^2-1)})node[anchor=north east]{\tiny $m_2$}; \fi

}

\draw[->]
({-0.2-1/(\d-1)},0) -- (1.2,0)node[anchor=north]{\small $s$} ;
\draw[->]
(0,{-0.2-1/(\d^2-1)}) -- (0,1.2)node[anchor=east]{\small $t$} ;

\draw (1,0)node[anchor=north west]{\tiny $1$};
\draw (0,1)node[anchor=south east]{\tiny $1$};

\end{tikzpicture}}
\caption{The set of all admissible points for the pair of identity channels $(I,I)$ (dark green square) compared with the analogue set for \subref{fig:squares_1} a pair of sharp meters $(\Qo,\Po)$; \subref{fig:squares_2} a pair $(\Qo,I)$ in which $\Qo$ is a sharp meter and $I$ is the identity channel. 
The light green areas are the differences between the two compared sets.\label{fig:squares}}
\end{figure}

Clearly, the compatibility regions $\CR{(\Qo,\Po)}$, $\CR{(I,I)}$ and $\CR{(\Qo,I)}$ described in Sections \ref{sec:joint} and \ref{sec:sequential} are the intersections of their extended versions with the product set $[0,1]\times [0,1]$. Moreover, the chain of inclusions
\begin{equation} 
\ECR{(I,I)}\subseteq\ECR{(\Qo,I)}\subseteq\ECR{(\Qo,\Po)}\,.
\end{equation}
holds for the same reason as in the standard case. 
However, while all compatibility regions always contain the extreme points $(0,0)$, $(1,0)$ and $(0,1)$ of the square $[0,1]\times [0,1]$, their extended versions may not contain any extreme point of the respective sets of admissible points.

Generalizing the statement of Theorem \ref{thm:QP+II} to the case of negative noise parameters is straightforward for the pair of identity channels $(I,I)$. However, it requires some care for the pairs of mutually unbiased sharp meters $(\Qo,\Po)$, as the case with $d=2$ is special.

\begin{theorem}\label{thm:QP+II_neg}
For $k=1$ or $k=2$, suppose $(s,t)\in [m_k,1]\times [m_k,1]$ and denote by $s\wedge t$ the minimum between $s$ and $t$.
\begin{enumerate}[(a)]
\item If $k=1$, for any two mutually unbiased sharp meters $\Qo$ and $\Po$, the meters $\Qo_s$ and $\Po_t$ are compatible if and only if either $d=2$ and\label{it:QP+II_neg_1}
\begin{equation}\label{eq:QP_comp_neg_1}
s^2+t^2\leq 1
\end{equation}
or $d\geq 3$ and
\begin{equation}\label{eq:QP_comp_neg_2}
s+t-1\leq \frac{2}{d}\left[s\wedge t - 1 + \sqrt{(s\wedge t)^2(1-d) + (s\wedge t)(d-2) + 1}\right]\,.
\end{equation}
\item If $k=2$, the channels $I_s$ and $I_t$ are compatible if and only if\label{it:QP+II_neg_2}
\begin{equation}\label{eq:II_comp_neg}
s+t-1\leq \frac{2}{d^2}\left[s\wedge t - 1 + \sqrt{(s\wedge t)^2(1-d^2) + (s\wedge t)(d^2-2) + 1}\right]\,.
\end{equation}
\end{enumerate}
\end{theorem}
The proof of item \eqref{it:QP+II_neg_1} was provided in \cite{CaHeTo19,CaCaTo19}, while that of item \eqref{it:QP+II_neg_2} easily follows by the results of \cite{Hashagen17}. In particular, item \eqref{it:QP+II_neg_1} implies that the reverse versions $\Qo_{m_1}$ and $\Po_{m_1}$ are compatible for all $d\geq 3$, a fact that was already observed in \cite{FiHeLe17} under the additional assumption that the meters $\Qo$ and $\Po$ are Fourier conjugate. The peculiarity of $d=2$ is explained by observing that, for a meter with only two outcomes, the reversing operation only swaps the outcomes but does not add any noise.

If $(s,t)\in [0,1]\times[0,1]$, replacing the minimum $s\wedge t$ by $s$ or $t$ does not affect the resulting conditions in \eqref{eq:QP_comp_neg_2} and \eqref{eq:II_comp_neg}. This substitution and some further manipulation show that, for nonnegative $s$ and $t$, the inequalities appearing in items \eqref{it:QP+II_neg_1} and \eqref{it:QP+II_neg_2} of Theorem \ref{thm:QP+II_neg} are consistent with those in the analogue items of Theorem \ref{thm:QP+II}.

To illustrate the differences arising for negative noise parameters, the extended compatibility regions characterized by Theorem \ref{thm:QP+II_neg} are depicted in Figure \ref{fig:comp_neg_QP+II}. As in the case with $s,t\geq 0$, also the regions given by \eqref{eq:QP_comp_neg_1} - \eqref{eq:II_comp_neg} are symmetric along the $s=t$ line. Moreover, in \eqref{eq:QP_comp_neg_2} (respectively, in \eqref{eq:II_comp_neg}) the equality is still attained if and only if the point $(s,t)$ lies on an arc segment of the ellipse \eqref{eq:poly}, namely the arc segment obtained by intersecting the half-plane $s+t\geq (d^k-3)/(d^k-1)$ and the ellipse \eqref{eq:poly} with $k=1$ (resp., with $k=2$). Despite these similarities, in the extended case a remarkable difference emerges between dimensions $d=2$ and $d\geq 3$.
Indeed, as we already observed, while in dimension $d\geq 3$ the point $(m_1,m_1)$ belongs to $\ECR{(\Qo,\Po)}$, this is no longer true in dimension $d=2$. Geometrically, the reason is that, when $d=2$, the set $\ECR{(\Qo,\Po)}$ is the unit disk centered at $(0,0)$, hence it is symmetric not only along the $s=t$ line, but also along any other line passing through the origin. 
In particular, $(m_1,m_1) = -(1,1)$ is outside $\ECR{(\Qo,\Po)}$ for $d=2$. Interestingly, no exception occurs if instead we consider the set $\ECR{(I,I)}$ in place of $\ECR{(\Qo,\Po)}$. In fact, $\ECR{(I,I)}$ contains the point $(m_2,m_2)$ for all $d\geq 2$ and it does not possess any additional symmetry for $d=2$.

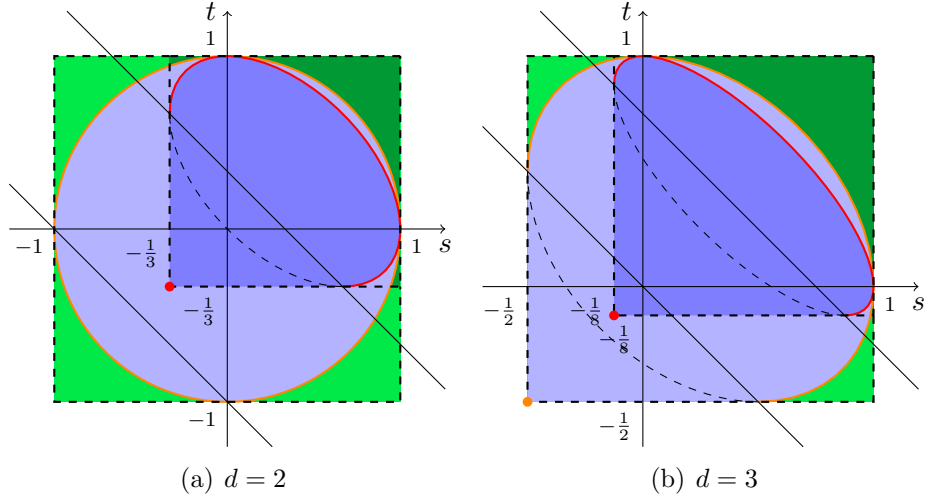
\begin{figure}[h!]
\centering
\subfigure[$d=2$]{\begin{tikzpicture}
\pic at (0,0) {figQPII=2/0/1/4.6};
\end{tikzpicture}}
\subfigure[$d=3$]{\begin{tikzpicture}
\pic at (0,0) {figQPII=3/0/1/4.6};
\end{tikzpicture}}
\caption{The extended compatibility regions $\ECR{(\Qo,\Po)}$ and $\ECR{(I,I)}$ given by Theorem \ref{thm:QP+II_neg} (blue areas) inside the respective sets of admissible points (green squares) in dimensions $d=2,3$.
The lighter blue area is the set difference $\ECR{(\Qo,\Po)}\setminus\ECR{(I,I)}$. Points of the orange (respectively, red) curve attain the equalities in \eqref{eq:QP_comp_neg_1}, \eqref{eq:QP_comp_neg_2} (resp., \eqref{eq:II_comp_neg}). The isolated colored points still belong to the boundary of $\ECR{(\Qo,\Po)}$ or $\ECR{(I,I)}$. For each $k=1,2$, we also plotted the line $s+t=(d^k-3)/(d^k-1)$ and the arc of the ellipse \eqref{eq:poly} which lies below it.\label{fig:comp_neg_QP+II}}
\end{figure}

It is natural to ask if the case with $d=2$ is special also for the pair $(\Qo,I)$ in which $\Qo$ is any sharp meter. The answer is affirmative, as illustrated by the next result.
\begin{theorem}\label{thm:main_neg}
Suppose $(s,t)\in [m_1,1]\times [m_2,1]$. For any sharp meter $\Qo$, the meter $\Qo_s$ and the channel $I_t$ are compatible if and only if
\begin{equation}\label{eq:QI_comp_neg}
\begin{aligned}
& - 2\left[1 - t_d + \sqrt{(1 - d^2) \smash[b]{t_d}^2 + (d^2 - 2) \smash[b]{t_d} + 1}\right] \leq d^2 s \\
& \qquad\qquad \leq (d^2 - 2) (1 - t) + 2\sqrt{(1 - d^2) t^2 + (d^2 - 2) t + 1} \,,
\end{aligned}
\end{equation}
where
\begin{equation} 
t_d = \begin{cases}
\displaystyle \frac{d - 2}{2(d - 1)} & \displaystyle \text{ if } d \geq 3 \text{ and } t\leq\frac{d - 2}{2(d - 1)}\,, \\[0.3cm]
t & \text{ otherwise}\,.
\end{cases}
\end{equation}
\end{theorem}

The above theorem characterizes the extended compatibility region of the pair $(\Qo,I)$. To describe the boundary of this region, we observe that, for $s\geq 0$, the second inequality of \eqref{eq:QI_comp_neg} is equivalent to inequality \eqref{eq:II_comp_neg}, hence both inequalities are saturated at the same points $(s,t)$ of the ellipse \eqref{eq:poly} with $k=2$. On the other hand, for $s<0$, the first inequality of \eqref{eq:QI_comp_neg} requires different treatments for $d=2$ and $d\geq 3$. Indeed, in all dimensions $d\geq 2$ the inequality is saturated if the point $(s,t)$ lies on a specific arc segment of the ellipse
\begin{equation}\label{eq:poly_neg}
d^2 s^2 + 4 t^2 + 4 (1 + s) (1 - t) = 4 \,.
\end{equation}
Namely, this arc segment is the one which lies in the half-plane $t - 2 s \geq 1$ for $d=2$, and in the half-plane $t - (d/2)s \geq 1$ for $d\geq 3$. When $d\geq 3$, however, there are additional points saturating the first inequality of \eqref{eq:QI_comp_neg}, i.e., all points $(s,t)$ with $s=m_1$ and $t\leq (d - 2)/[2 (d - 1)]$. It follows that the point $(m_1,m_2)$ belongs to $\ECR{(\Qo,I)}$ if and only if $d\geq 3$. Moreover, the set $\ECR{(\Qo,I)}$ is symmetric along the $s=0$ line for $d=2$, while it does not possess any symmetry for $d\geq 3$. This discussion is summarized in Figure \ref{fig:comp_neg_QI}, which depicts the set $\ECR{(\Qo,I)}$ characterized by Theorem \ref{thm:QP+II_neg} and compares it with the set $\ECR{(I,I)}$.

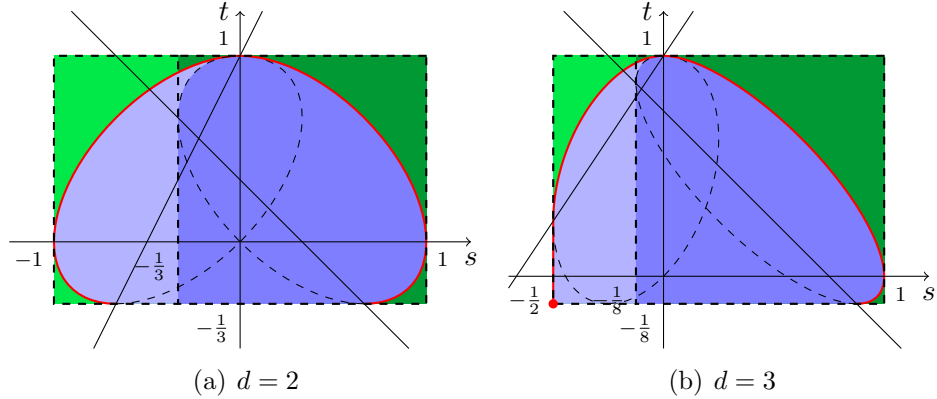
\begin{figure}[h!]
\centering
\subfigure[$d=2$\label{fig:comp_neg_QI_1}]{\begin{tikzpicture}
\pic at (0,0) {figQI=2/0/3.3};
\end{tikzpicture}}
\subfigure[$d=3$\label{fig:comp_neg_QI_2}]{\begin{tikzpicture}
\pic at (0,0) {figQI=3/0/3.3};
\end{tikzpicture}}
\caption{The extended compatibility region $\ECR{(\Qo,I)}$ given by Theorem \ref{thm:main_neg} (blue area) inside the set of all admissible points for the pair $(\Qo,I)$ (green rectangle) in dimensions $d=2,3$.
The lighter colors denote the sets $\ECR{(\Qo,I)}\setminus\ECR{(I,I)}$ and $[m_1,m_2]\times [m_2,1]$. Points of the red curve attain the equalities in \eqref{eq:QI_comp_neg} and the point with coordinates $(m_1,m_2)$ is highlighted. Of the two lines which appear in both graphs, one has equation $s+t=(d^2-3)/(d^2-1)$, while the other one is $t-2s=1$ in graph \subref{fig:comp_neg_QI_1} and $t-(d/2)s=1$ in graph \subref{fig:comp_neg_QI_2}. We also plotted the arcs of the ellipses \eqref{eq:poly} and \eqref{eq:poly_neg} which lie below these lines.\label{fig:comp_neg_QI}}
\end{figure}

As a consequence of Theorems \ref{thm:QP+II_neg} and \ref{thm:main_neg}, contrary to what happens with nonnegative noise parameters, now the inclusion $\ECR{(I,I)}\subseteq\ECR{(\Qo,I)}$ is strict in all dimensions $d\geq 2$. In particular, for $s\in [m_2,0)$, we can find $t\in (0,1)$ such that $\Qo_s$ and $I_t$ are compatible, while $I_s$ and $I_t$ are not. The next main result is immediately implied by this fact.
\begin{result}
If negative values of the noise parameters are allowed, Result \ref{res:1} is no longer true. Namely, the strategy consisting in performing a sequential measurement that is optimized for a specific choice of $\Qo$ can provide an advantage over the optimal cloning strategy.
\end{result}
This result is just what one would conjecture for all values of $s$ and $t$ (positive or negative). In view of it, the equality $\CR{(\Qo,I)} = \CR{(I,I)}$ found in Section \ref{sec:sequential} is even more surprising and unexpected.

\subsection{Optimal joint devices and Result \ref{res:2} for negative noise parameters}\label{sec:ext2}

We now describe the approximate joint devices (meters, channels or instruments) which yield the extreme points of the extended compatibility regions discussed above. We preliminary observe that the quantities $a_k(r)$ and $b(r)$ introduced in \eqref{eq:ab_bis} for $k=1,2$ and $r\in [0,1]$ can be defined also for $r\in [m_k,0)$, and they still satisfy the identity \eqref{eq:ab_ellipse} for these further values of $r$. The main difference arising for negative noise parameters is that we have $a_k(r)<0$ and $b(r)>1$ for $r<0$. Moreover, when $s,t\in [m_k,1]$, the equivalences \eqref{eq:a=b_equiv} are no longer true and they need to be replaced by
\begin{equation}\label{eq:a=b_equiv_neg}
\begin{aligned}
& |a_k(s\wedge t)| = b(s\vee t) \quad\Leftrightarrow\quad \\
& \quad\Leftrightarrow\quad
\begin{aligned}[t]
& s+t-1 = \\
& \quad = \frac{2}{d^k}\left[s\wedge t - 1 + \sqrt{(s\wedge t)^2(1-d^k) + (s\wedge t)(d^k-2) + 1}\right] \,,
\end{aligned}
\end{aligned}
\end{equation}
where $s\vee t$ denotes the maximum between $s$ and $t$. In the particular case with $s<0$ or $t<0$, the equality $|a_k(s\wedge t)| = b(s\vee t)$ amounts to $-a_k(s\wedge t) = b(s\vee t)$, and we have the implication 
\begin{equation}\label{eq:a=b_equiv_neg_neg}
\begin{aligned}
& -a_k(s\wedge t) = b(s\vee t) \quad\Rightarrow\quad a_k(s\vee t) + \frac{2}{\sqrt{d^k}} b(s\vee t) = b(s\wedge t)\,.
\end{aligned}
\end{equation}

Extending $a_k(r)$ and $b(r)$ to negative values of $r$ has the following consequences on the approximate joint devices described in \eqref{eq:optGo} and \eqref{eq:optGamma}. Firstly, the meter $\Go^{(s)}$ of \eqref{eq:optGo} is well defined also for $s\in [m_1,0)$, since the square root of $\Qo_s(x)$ is still given by \eqref{eq:Q1/2}. Then, assuming that $\Qo$ and $\Po$ are mutually unbiased, $\Go^{(s)}$ is again a joint meter for $\Qo_{1-b(s)^2} = \Qo_s$ and $\Po_{1-a_1(s)^2}$ by the same argument as for \eqref{eq:optGo_margins}. However, when an admissible point $(s,t)$ attains the equality in \eqref{eq:QP_comp_neg_2}, the equivalence \eqref{eq:a=b_equiv_neg} does not entail that $\Go^{(s)}$ is a joint meter for $\Qo_s$ and $\Po_t$ unless the additional condition $t\geq 0$ is satisfied. Indeed, $1-a_1(s)^2\geq 0$ for all $s$ (positive or negative), hence $\Po_{1-a_1(s)^2}$ cannot coincide with $\Po_t$ if $t<0$. By a similar reason, in order that \eqref{eq:optGamma} defines a joint channel $\Gamma^{(s)}$ for $I_s$ and $I_t$, it is not enough that $(s,t)$ is an admissible point which attains the equality in \eqref{eq:II_comp_neg}, and the inequality $t\geq 0$ is also required. To derive joint meters and channels that are optimal for $t<0$, it suffices to replace $a_k(s)$ and $b(s)$ with $b(t)$ and $a_k(t)$, respectively, in the formulas for $\Go^{(s)}$ and $\Gamma^{(s)}$. This substitution yields the meter $\tilde{\Go}^{(t)}$ given by
\begin{equation}\label{eq:tildeGo}
\begin{aligned}
& \tilde{\Go}^{(t)}(x,y) = \Po_t(y)^{\frac{1}{2}}\Qo(x)\Po_t(y)^{\frac{1}{2}} \\
& \quad = \frac{1}{d} b(t)^2 \Qo(x) + \frac{1}{d} a_1(t)^2 \Po(y) + \frac{1}{\sqrt{d}} a_1(t) b(t) \big(\Qo(x)\Po(y) + \Po(y)\Qo(x)\big)
\end{aligned}
\end{equation}
for all $t\in [m_1,1]$, and the channel $\tilde{\Gamma}^{(t)}$ given by
\begin{equation}\label{eq:tildeGamma}
\tilde{\Gamma}^{(t)}(\rho) = \big(b(t)\id\otimes\id + a_2(t)F\big)\left(\rhoo\otimes\tfrac{1}{d}\id\right)\big(b(t)\id\otimes\id + a_2(t)F\big)
\end{equation}
for all $t\in [m_2,1]$. For $s,t\geq 0$, conditions \eqref{eq:a=b_equiv} and \eqref{eq:a=b_equiv_neg} are equivalent, hence $\Go^{(s)} = \tilde{\Go}^{(t)}$ (respectively, $\Gamma^{(s)} = \tilde{\Gamma}^{(t)}$) for all points $(s,t)\in [0,1]\times [0,1]$ which attain the equality in \eqref{eq:QP_comp} (resp., in \eqref{eq:II_comp}).

In order to conclude our description of the optimal approximate joint meters and channels in the extended case, we finally consider noise parameters $s<0$ and $t<0$, which, as we have seen, give different behaviours in dimension $d=2$. Indeed, if $d=2$, then
\begin{equation}\label{eq:Gominus}
\Go^{(s)}_-(x,y) = \Qo_s(x)^{\frac{1}{2}}\Po_{-1}(y)\Qo_s(x)^{\frac{1}{2}} = \Qo_s(x) - \Go^{(s)}(x,y)
\end{equation}
defines a joint meter $\Go^{(s)}_-$ for $\Qo_s$ and $\Po_{-(1-a(s)^2)}$ for all $s\in [-1,1]$. In particular, $\Go^{(s)}_-$ is a joint meter for $\Qo_s$ and $\Po_t$ when the point $(s,t)$ attains the equality in \eqref{eq:QP_comp_neg_1} and $t\leq 0$. In the latter case one can further show that the equality $\Go^{(s)}_- = \tilde{\Go}^{(t)}$ holds for $s\geq 0$ by applying \eqref{eq:a=b_equiv_neg_neg} and the relation $\Qo(x)\Po(y) + \Po(y)\Qo(x) = \Qo(x) + \Po(y) - \frac{1}{2}\id$. On the other hand, if $d\geq 3$, then a joint meter $\Go_\Lcorner$ for the reverse versions $\Qo_{m_1}$ and $\Po_{m_1}$ is given by
\begin{equation}\label{eq:Gocorner}
\Go_\Lcorner(x,y) = \frac{1}{d(d-2)}\left[\id - \frac{d}{d-1}\big(\Qo(x)-\Po(y)\big)^2\right]\,.
\end{equation}
Note that the expression $[d/(d-1)]\big(\Qo(x)-\Po(y)\big)^2$ is the projection onto the linear span of the basis elements $\phii_x$ and $\psi_y$ (see \eqref{eq:QP}). When we consider channels in place of meters, the analogue of the joint meter $\Go_\Lcorner$ is the channel $\Gamma_\Lcorner$ given by
\begin{equation}\label{eq:Gammacorner}
\Gamma_\Lcorner(\rhoo) = \frac{1}{d^2-2}\bigg\{\tr{\rhoo}\id\otimes\id - 2\sum_{h=\pm 1} \frac{1}{d+h} S_h(\rhoo\otimes\id) S_h\bigg\} \,,
\end{equation}
where $S_{+1} = \frac{1}{2}(\id\otimes\id + F)$ (respectively, $S_{-1} = \frac{1}{2}(\id\otimes\id - F)$) is the projection onto the symmetric (resp., antisymmetric) subspace of $\Cb^d\otimes\Cb^d$. Indeed, we prove in Appendix \ref{app:Gamma_J} that $\Gamma_\Lcorner$ is a joint channel for two copies of the channel $I_{m_2}$.
The channel $\Gamma_\Lcorner$ is defined for all $d\geq 2$, with no difference between $d=2$ and $d\geq 3$.

In Table \ref{ta:QP+II} we provide a summary of the approximate joint devices described above and we highlight those points they correspond to in the extended compatibility regions.

\begin{table}[h!]\def\vsp{0.1cm}\def\scala{0.35}
\begin{tabular}{ || >{\centering\arraybackslash}m{0.7cm} | >{\centering\arraybackslash}m{2.2cm} >{\centering\arraybackslash}m{2.2cm} || >{\centering\arraybackslash}m{0.7cm} | >{\centering\arraybackslash}m{2.2cm} >{\centering\arraybackslash}m{2.2cm} || }
\toprule
\multicolumn{3}{||c||}{\textbf{Q \& P}} & \multicolumn{3}{c||}{\textbf{I \& I}}
\\ \cmidrule(rl){1-3} \cmidrule(rl){4-6}
& {$d=2$} & {$d\geq 3$} & & {$d=2$} & {$d\geq 3$}
\\ \midrule
\parbox[c]{\hsize}{\centering $\Go^{(s)}$ \\[\vsp] \eqref{eq:optGo}} &
\scalebox{\scala}{\begin{tikzpicture}
\pic at (0,0) {figQPII=2/1/1/4.6};
\end{tikzpicture}} &
\scalebox{\scala}{\begin{tikzpicture}
\pic at (0,0) {figQPII=3/1/1/4.6};
\end{tikzpicture}} &
\parbox[c]{\hsize}{\centering $\Gamma^{(s)}$ \\[\vsp] \eqref{eq:optGamma}} &
\scalebox{\scala}{\begin{tikzpicture}
\pic at (0,0) {figQPII=2/1/2/4.6};
\end{tikzpicture}} &
\scalebox{\scala}{\begin{tikzpicture}
\pic at (0,0) {figQPII=3/1/2/4.6};
\end{tikzpicture}}
\\ \midrule
\parbox[c]{\hsize}{\centering $\tilde{\Go}^{(t)}$ \\[\vsp] \eqref{eq:tildeGo}} &
\scalebox{\scala}{\begin{tikzpicture}
\pic at (0,0) {figQPII=2/2/1/4.6};
\end{tikzpicture}} &
\scalebox{\scala}{\begin{tikzpicture}
\pic at (0,0) {figQPII=3/2/1/4.6};
\end{tikzpicture}} &
\parbox[c]{\hsize}{\centering $\tilde{\Gamma}^{(t)}$ \\[\vsp] \eqref{eq:tildeGamma}} &
\scalebox{\scala}{\begin{tikzpicture}
\pic at (0,0) {figQPII=2/2/2/4.6};
\end{tikzpicture}} &
\scalebox{\scala}{\begin{tikzpicture}
\pic at (0,0) {figQPII=3/2/2/4.6};
\end{tikzpicture}}
\\ \midrule
\parbox[c]{\hsize}{\centering $\Go^{(s)}_-$ \\[\vsp] \eqref{eq:Gominus}} &
\scalebox{\scala}{\begin{tikzpicture}
\pic at (0,0) {figQPII=2/3/1/4.6};
\end{tikzpicture}} &
&
&
&
\\ \midrule
\parbox[c]{\hsize}{\centering $\Go_\Lcorner$ \\[\vsp] \eqref{eq:Gocorner}} &
&
\scalebox{\scala}{\begin{tikzpicture}
\pic at (0,0) {figQPII=3/3/1/4.6};
\end{tikzpicture}} &
\parbox[c]{\hsize}{\centering $\Gamma_\Lcorner$ \\[\vsp] \eqref{eq:Gammacorner}} &
\scalebox{\scala}{\begin{tikzpicture}
\pic at (0,0) {figQPII=2/3/2/4.6};
\end{tikzpicture}} &
\scalebox{\scala}{\begin{tikzpicture}
\pic at (0,0) {figQPII=3/3/2/4.6};
\end{tikzpicture}}
\\ \bottomrule
\end{tabular}\bigskip
\caption{The extreme points of the extended compatibility regions $\ECR{(\Qo,\Po)}$ (left) and $\ECR{(I,I)}$ (right) that correspond to the approximate joint meters and channels described in Section \ref{sec:ext2}. Each figure refers to the joint device reported on its left and the dimension indicated above. The number between parentheses under each joint device specifies the respective equation in the text.\label{ta:QP+II}}
\end{table}

We now pass to the description of the approximate joint instruments which yield the extreme points of $\ECR{(\Qo,I)}$. By means of the cloning device $\tilde{\Gamma}^{(t)}$, we can immediately obtain a joint instrument $\tilde{\ii}^{(t)}$ for $\Qo_s$ and $I_t$ when $(s,t)$ is any admissible point which attains the second equality in \eqref{eq:QI_comp_neg}. Indeed, it is enough to consider the following modification of the optimal instrument $\ii^{(s)}$ of \eqref{eq:optJ}
\begin{equation}\label{eq:tildeJ}
\begin{aligned}
& \tilde{\ii}^{(t)}(x,\rhoo) = \ptr{\bar{2}}{\left[(\Qo(x)\otimes\id)\tilde{\Gamma}^{(t)}(\rhoo)\right]} \\
& \quad = \frac{1}{d} \big\{ b(t)^2 \tr{(\rhoo\Qo(x))} \id + a_2(t)b(t) (\rhoo\Qo(x) + \Qo(x)\rhoo) + a_2(t)^2 \rhoo \big\} \,.
\end{aligned}
\end{equation}
In the particular case with $(s,t)\in [0,1]\times [0,1]$, the equality $\Gamma^{(s)} = \tilde{\Gamma}^{(t)}$ entails that $\ii^{(s)} = \tilde{\ii}^{(t)}$. On the other hand, when $(s,t)$ attains the first equality in \eqref{eq:QI_comp_neg}, it is easy to check that the instrument $\tilde{\ii}^{(t)}_-$ given by
\begin{equation}\label{eq:tildeJminus}
\begin{aligned}
& \tilde{\ii}^{(t)}_-(x,\rhoo) = \frac{1}{d}\bigg\{ b(t)^2 \tr{\rhoo(\id-\Qo(x))}\Qo(x) \\
& \ \ + \bigg[\bigg(a_2(t)+\frac{2}{d}b(t)\bigg)\id-b(t)\Qo(x)\bigg]\rhoo\bigg[\bigg(a_2(t)+\frac{2}{d}b(t)\bigg)\id-b(t)\Qo(x)\bigg]\bigg\}
\end{aligned}
\end{equation}
is a joint instrument for $\Qo_s$ and $I_t$ for all $t\geq t_d$. Indeed, by using \eqref{eq:ab_ellipse}, a straightforward calculation leads to
\begin{equation} 
\tilde{\ii}^{(t)\,\dag}_-(\cdot,\id) = \Qo_{-\frac{2}{d}b(t)\left(a_2(t)+\frac{2}{d}b(t)\right)}\,,\qquad \sum_x \tilde{\ii}^{(t)}_-(x,\cdot) = I_{1-b(t)^2}\,,
\end{equation}
where $1-b(t)^2 = t$ and
\begin{equation*}
-\frac{2}{d}b(t)\bigg(a_2(t)+\frac{2}{d}b(t)\bigg) = -\frac{2}{d^2} \left[ 1 - t + \sqrt{(1 - d^2) t^2 + (d^2 - 2) t + 1}\right]\,.
\end{equation*}
The instruments $\tilde{\ii}^{(t)}$ and $\tilde{\ii}^{(t)}_-$ yield all the extreme points of $\ECR{(\Qo, I)}$ when $d=2$. On the other hand, when $d\geq 3$, we still need a joint instrument $\ii_\Lcorner$ yielding the extreme point $(m_1,m_2)$. 
We can choose
\begin{equation}\label{eq:Jcorner}
\begin{aligned}
& \ii_\Lcorner(x,\rhoo) = \frac{1}{d^2-1} \bigg\{ \frac{1}{d-2}\big(\Qo(x)\rhoo+\rhoo\Qo(x)-\rhoo\big) \\
& \qquad + \frac{1}{(d-1)(d-2)}\bigg[ \sum_z\tr{\rhoo\Qo(z)}\Qo(z) - d\tr{\rhoo\Qo(x)}\Qo(x)\bigg] \\
& \qquad + \frac{d}{d-1}\tr{\rhoo\big(\id-\Qo(x)\big)}\id\bigg\}\,,
\end{aligned}
\end{equation}
which is indeed a joint instrument for $\Qo_{m_1}$ and $I_{m_2}$ (see the next section).
Table \ref{ta:QI} summarizes all the above approximate joint instruments and highlights the corresponding points of $\ECR{(\Qo,I)}$.

\begin{table}[h!]\def\vsp{0.1cm}\def\scala{0.33}
\begin{tabular}{ || >{\centering\arraybackslash}m{0.72cm} | >{\centering\arraybackslash}m{2.1cm} >{\centering\arraybackslash}m{2.1cm} || >{\centering\arraybackslash}m{0.72cm} | >{\centering\arraybackslash}m{2.1cm} >{\centering\arraybackslash}m{2.1cm} || }
\toprule
\multicolumn{6}{||c||}{\textbf{Q \& I}}
\\ \cmidrule(rl){1-6}
& {$d=2$} & {$d\geq 3$} & & {$d=2$} & {$d\geq 3$}
\\ \midrule
\parbox[c]{\hsize}{\centering $\ii^{(s)}$ \\[\vsp] \eqref{eq:optJ}} &
\scalebox{\scala}{\begin{tikzpicture}
\pic at (0,0) {figQI=2/1/3.3};
\end{tikzpicture}} &
\scalebox{\scala}{\begin{tikzpicture}
\pic at (0,0) {figQI=3/1/3.3};
\end{tikzpicture}} &
\parbox[c]{\hsize}{\centering $\tilde{\ii}^{(t)}$ \\[\vsp] \eqref{eq:tildeJ}} &
\scalebox{\scala}{\begin{tikzpicture}
\pic at (0,0) {figQI=2/2/3.3};
\end{tikzpicture}} &
\scalebox{\scala}{\begin{tikzpicture}
\pic at (0,0) {figQI=3/2/3.3};
\end{tikzpicture}}
\\ \midrule
\parbox[c]{\hsize}{\centering $\tilde{\ii}^{(t)}_-$ \\[\vsp] \eqref{eq:tildeJminus}} &
\scalebox{\scala}{\begin{tikzpicture}
\pic at (0,0) {figQI=2/3/3.3};
\end{tikzpicture}} &
\scalebox{\scala}{\begin{tikzpicture}
\pic at (0,0) {figQI=3/3/3.3};
\end{tikzpicture}} &
\parbox[c]{\hsize}{\centering $\ii_\Lcorner$ \\[\vsp] \eqref{eq:Jcorner}} &
\scalebox{\scala}{\begin{tikzpicture}
\pic at (0,0) {figQI=2/4/3.3};
\end{tikzpicture}} &
\scalebox{\scala}{\begin{tikzpicture}
\pic at (0,0) {figQI=3/4/3.3};
\end{tikzpicture}}
\\ \bottomrule
\end{tabular}\bigskip
\caption{The extreme points of the extended compatibility region $\ECR{(\Qo,I)}$ that correspond to the approximate joint instruments described in Section \ref{sec:ext2}. The table reads as already described for Table \ref{ta:QP+II}.\label{ta:QI}}
\end{table}

To conclude this section, we show that the approximate joint meters $\Go^{(s)}$, $\tilde{\Go}^{(t)}$, $\Go^{(s)}_-$ and $\Go_\Lcorner$ described above can be obtained in a sequential measurement scheme in which we first perform a fixed measurement of a (over)noisy version of $\Qo$ -- the same measurement for all $\Po$'s --  and subsequently a measurement of $\Po$. In this way, we prove that Result \ref{res:2} is still true even if negative values of the noise parameters are allowed. This fact is clear for the meter $\Go^{(s)}$, as the L\"uders instrument $\ii_L^{(s)}$ of \eqref{eq:Luders} is well defined and satisfies \eqref{eq:Luders+P} also for $s\in [m_1,0)$. For all the other meters, we can still use $\ii_L^{(s)}$ in combination with a suitably chosen state postprocessing. More precisely, conditionally to the outcome $x$ of $\Qo_s$, we can transform the posterior state $\ii_L^{(s)}(x,\rhoo)/\tr{\vphantom{\big(}\smash{\ii_L^{(s)}(x,\rhoo)}}$ by means of a channel $\Pi(\cdot\,|\,x) :\lc{d}\to\lc{d}$, thus obtaining the overall instrument $\ii_{L,\Phi}^{(s)}$ given by
\begin{equation} 
\ii_{L,\Phi}^{(s)}(x,\rhoo) = \Pi\big(\ii_L^{(s)}(x,\rhoo)\,|\,x\big)\,.
\end{equation}

By setting
\begin{equation} 
\Pi(\rhoo\,|\,x) = (\id-2\Qo(x))\rhoo(\id-2\Qo(x))
\end{equation}
for all states $\rhoo$ and using the identity \eqref{eq:ab_ellipse} and the mutual unbiasedness of $\Qo$ and $\Po$, a straightforward calculation yields
\begin{align*}
\ii_{L,\Phi}^{(s)\,\dag}(x,\Po(y)) = {}& \frac{1}{d} \bigg( a_1(s) + \frac{2}{\sqrt{d}} b(s) \bigg)^2 \Qo(x) + \frac{1}{d} b(s)^2 \Po(y) \\
& - \frac{1}{\sqrt{d}} b(s)\bigg( a_1(s) + \frac{2}{\sqrt{d}} b(s) \bigg) \big(\Qo(x)\Po(y) + \Po(y)\Qo(x)\big)
\end{align*}
for all $s\in [m_1,1]$. Given any $t\in [m_1,1]$, by picking $s\geq 0$ such that the point $(s,t)$ attains the equality in \eqref{eq:QP_comp_neg_2}, the relations \eqref{eq:a=b_equiv_neg} and \eqref{eq:a=b_equiv_neg_neg} then yield
\begin{equation} 
\ii_{L,\Phi}^{(s)\,\dag}(x,\Po(y)) = \tilde{\Go}^{(t)}(x,y)
\end{equation}
whenever $t<0$. On the other hand, we already observed that $\Go^{(s)} = \tilde{\Go}^{(t)}$ for $t\geq 0$. In both cases the meter $\tilde{\Go}^{(t)}$ is obtained in a sequential measurement scheme in which we first perform a fixed measurement of $\Qo_s$ and subsequently a measurement of $\Po$, as claimed.

In dimension $d=2$ the further relation
$$
(\id-2\Qo(x))\Po(y)(\id-2\Qo(x)) = \Po_{-1}(y)
$$
holds for all $x,y$. This implies that
\begin{equation} 
\ii_{L,\Phi}^{(s)\,\dag}(x,\Po(y)) = \Go^{(s)}_-(x,y)
\end{equation}
for all $s\in [-1,1]$, hence also the meter $\Go^{(s)}_-$ is obtained in a sequential measurement scheme as above.

To finally show that, when $d\geq 3$, the sequential measurement scheme can be implemented also for the remaining meter $\Go_\Lcorner$, we need to modify the definition of the conditional channel $\Pi(\cdot\,|\,x)$ and set
\begin{equation} 
\Pi(\rhoo\,|\,x) = \frac{1}{d-2}\sum_z (1-\delta_{x,z}) \bigg(\frac{1}{d-1}\id - \Qo(z)\bigg)\rhoo\bigg(\frac{1}{d-1}\id - \Qo(z)\bigg) \,,
\end{equation}
where $\delta_{x,z}$ is the Kronecker delta. With the latter choice of $\Pi(\cdot\,|\,x)$ a straightforward application of the mutual unbiasedness of $\Qo$ and $\Po$ yields
\begin{equation} 
\ii_{L,\Phi}^{(m_1)\,\dag}(x,\Po(y)) = \Go_\Lcorner(x,y)\,.
\end{equation}
This relation completes the proof of the final main result of the paper.
\begin{result}
Result \ref{res:2} still holds true also if negative values of the noise parameters are allowed. Namely, when $\Po$ is known to be unbiased with respect to $\Qo$, the strategy consisting in performing a sequential measurement that is optimized for those pairs $(\Qo,\Po)$ in which only $\Qo$ is fixed overcomes the cloning strategy and yields the same trade-off as measuring an optimal approximate joint meter of a specific pair $(\Qo,\Po)$.
\end{result}

\section{Proof of the main results}\label{sec:proof}

In this section we prove our main Theorems \ref{thm:main} and \ref{thm:main_neg} by means of a simple symmetrization trick and few elementary facts from the theory of Weyl-Heisenberg groups and Fourier analysis. Before doing it, we need to introduce some further material about finite dimensional Weyl operators and the related covariance properties, see e.g.~\cite{AuTo79,BaIt86,StTo84,Va95,Te99} for more details on the topic.

We start by observing that one can define a natural group structure on the outcome sets $X$ and $Y$ of the two sharp meters $\Qo$ and $\Po$ defined in \eqref{eq:QP}. Indeed, both sets can be regarded as the cyclic group $\Zb_d$, i.e., the group of integer numbers $\{0,1,\ldots,d-1\}$ in which the composition law is given by addition ${\rm mod}\,d$. If $\omega$ denotes any fixed complex $d$th root of unity and $x\in\Zb_d$, the map $z\mapsto\omega^{xz}$ is then a group homomorphism from $\Zb_d$ into the multiplicative group of unimodular complex numbers. For all $x,y\in\Zb_d$, the {\em Weyl operator} $W(x,y)\in\lc{d}$ is defined as
\begin{equation} 
W(x,y)\phii_z = \omega^{yz} \phii_{x+z}
\end{equation}
for all $z\in\Zb_d$. In this expression $\{\phii_0,\ldots,\phii_{d-1}\}$ is the orthonormal basis associated with $\Qo$. The Weyl operators are unitary and satisfy the composition rule
$$
W(x_1,y_1)W(x_2,y_2) = \omega^{x_2 y_1} W(x_1 + x_2, y_1 + y_2) \,.
$$
Therefore, the map $W:\Zb_d\times\Zb_d\to\lc{d}$ is a projective unitary representation of the direct product group $\Zb_d\times\Zb_d$. Moreover, we have the commutation relation
$$
W(x_1,y_1)W(x_2,y_2) = \omega^{x_2 y_1 - x_1 y_2} W(x_2,y_2)W(x_1,y_1)\,.
$$
For all $s\in [m_1,1]$ and $t\in [m_2,1]$, the interplay between the Weyl operators and the meters $\Qo_s$ and the channels $I_t$ defined in the previous sections is given by the following covariance and invariance relations
\begin{align}
W(x,y)\Qo_s(z)W(x,y)^* & = \Qo_s(x+z)\,, \label{eq:cov_Qs}\\
W(x,y)I_t(\rhoo)W(x,y)^* & = I_t(W(x,y)\rhoo W(x,y)^*)\,, \label{eq:cov_It}
\end{align}
which hold for all $x,y,z$ and $\rho$.

\begin{lemma}[Symmetrization trick]\label{lem:trick}
Suppose $\Qo_s$ and $I_t$ are compatible. Then, there exists a joint instrument $\ii$ for $\Qo_s$ and $I_t$ such that
\begin{equation}\label{eq:Icov_def}
W(x,y)\ii(z,\rhoo)W(x,y)^* = \ii(x+z,W(x,y)\rhoo W(x,y)^*)
\end{equation}
for all $x,y,z$ and $\rhoo$.
\end{lemma}
\begin{proof}
Let $\ii'$ be any joint instrument for $\Qo_s$ and $I_t$. For all $z$ and $\rho$, we define
$$
\ii(z,\rho) = \frac{1}{d^2} \sum_{x,y} W(x,y)\ii'(z-x,W(x,y)^*\rhoo W(x,y))W(x,y)^*\,.
$$
It is easy to check that $\ii$ is an instrument which satisfies \eqref{eq:Icov_def}. Moreover, $\ii$ is a joint instrument for $\Qo_s$ and $I_t$, since
\begin{align*}
\ii^\dag(z,\id) & = \frac{1}{d^2} \sum_{x,y} W(x,y)\ii^{\prime\,\dag}(z-x,W(x,y)^*\id W(x,y))W(x,y)^* \\
& = \frac{1}{d^2} \sum_{x,y} W(x,y)\Qo_s(z-x)W(x,y)^* = \Qo_s(z)\,, \\
\sum_z\ii(z,\rhoo) & = \frac{1}{d^2} \sum_{x,y} W(x,y) \sum_z \ii'(z-x,W(x,y)^*\rhoo W(x,y))W(x,y)^* \\
& = \frac{1}{d^2} \sum_{x,y} W(x,y)I_t(W(x,y)^*\rhoo W(x,y))W(x,y)^* = I_t(\rhoo)\,,
\end{align*}
where the second equalities follow from \eqref{eq:cov_Qs} and \eqref{eq:cov_It}.
\end{proof}

Instruments satisfying \eqref{eq:Icov_def} are called {\em $W$-covariant} and they were extensively studied in \cite{CaHeTo11}. In particular, \cite[Theorem 1]{CaHeTo11} establishes a one-to-one correspondence between $W$-covariant instruments $\ii$ and {\em positive $\lc{d}$-valued vector measures on $\Zb_d$}, i.e., maps $S:\Zb_d\to\lc{d}$ such that $\sum_x \tr{S(x)} = 1$ and $S(x)\geq 0$ for all $x$. This correspondence is given by
\begin{equation}\label{eq:Icov_char_1}
\ii(z,\rho) = \sum_x W(x,0)^*\ptr{\bar{1}}{[(\id\otimes\Qo(z))L(\rhoo\otimes S(x))L^*]}W(x,0)\,,
\end{equation}
where $L:\Cb^d\otimes\Cb^d\to\Cb^d\otimes\Cb^d$ is the unitary operator
\begin{equation}\label{eq:L}
L = \sum_u \Qo(u)\otimes W(u,0)\,.
\end{equation}
By inserting \eqref{eq:L} into \eqref{eq:Icov_char_1} and using the cyclicity of ${\rm tr}_{\bar{1}}$ with respect to the second factor in the tensor product, we find the equivalent expression
\begin{equation} 
\begin{aligned}
\ii(z,\rho) & = \sum_{x,u,v} W(x,0)^*\Qo(u)\rhoo\Qo(v)W(x,0) \\
& \qquad\times\tr{W(v,0)^*\Qo(z)W(u,0)S(x)}\,.
\end{aligned}
\end{equation}
From this equation, it is easy to show that $\ii$ is a joint instrument for the meter $\Mo$ and the channel $\Phi$ which are given by
\begin{align}
\Mo(z) & = \sum_u \tr{\vphantom{\bigg[}\Qo(z-u)\smash{\sum_x} S(x)}\Qo(u)\,,\label{eq:MI} \\
\Phi(\rhoo) & = \sum_{x,u,v} W(x,0)^*\Qo(u)\rhoo\Qo(v)W(x,0) \tr{W(u-v,0)S(x)}\label{eq:PhiI}
\end{align}
for all $z$ and $\rhoo$ (actually, the derivation of $\Mo$ is already contained in \cite[Proposition 1]{CaHeTo11}). By \eqref{eq:MI}, the meter $\Mo$ is a {\em smeared} (or {\em fuzzy}) {\em version} of the sharp meter $\Qo$, as it is obtained by convolving $\Qo$ with the probability distribution $q$ given by
\begin{equation}\label{eq:q}
q(z) = \tr{\vphantom{\bigg[}\Qo(z)\smash{\sum_x} S(x)}
\end{equation}
for all $z$ (see e.g.~\cite{HeLaYl04}). 
On the other hand, formula \eqref{eq:PhiI} can still be simplified by introducing the finite Fourier transform $\ff:\Cb^d\to\Cb^d$ and the orthonormal basis $\{\psi_0,\ldots,\psi_{d-1}\}$ which is Fourier conjugate to $\{\phii_0,\ldots,\phii_{d-1}\}$. They are defined by
\begin{equation} 
\ff \phii_z = \frac{1}{\sqrt{d}} \sum_u \omega^{-uz} \phii_u\,,\qquad \psi_z = \ff^* \phii_z
\end{equation}
for all $z$. In particular, the two bases $\{\phii_0,\ldots,\phii_{d-1}\}$ and $\{\psi_0,\ldots,\psi_{d-1}\}$ are mutually unbiased. By applying the relation
\begin{equation}\label{eq:ort_rel}
\sum_u \omega^{uv} = d\delta_{v,0}\,,
\end{equation}
it follows that
$$
W(-y,x) = \omega^{xy} \ff^* W(x,y)\ff
$$
and
\begin{equation}\label{eq:QW}
\begin{aligned}
W(0,y) & = \sum_z \omega^{yz} \Qo(z)\,,\qquad \Qo(z) = \frac{1}{d} \sum_y \omega^{-yz} W(0,y)\,, \\
W(x,0) & = \sum_z \omega^{-xz} \Po(z)\,,\qquad \Po(z) = \frac{1}{d} \sum_x \omega^{xz} W(x,0)\,,
\end{aligned}
\end{equation}
where $\Po$ is the sharp meter that is Fourier conjugate to $\Qo$, i.e.,
\begin{equation} 
\Po(z) = \kb{\psi_z}{\psi_z} = \ff^*\Qo(z)\ff\,.
\end{equation}
By inserting \eqref{eq:QW} in \eqref{eq:PhiI} and using again \eqref{eq:ort_rel}, we obtain
\begin{equation} 
\Phi(\rhoo) = \sum_{x,z} W(x,z)^*\rhoo W(x,z) \tr{\Po(z)S(x)}\,.
\end{equation}
In this expression we recognize the Kraus form of a {\em Weyl-covariant channel}. In particular, $\Phi$ is the Weyl-covariant channel associated with the probability distribution $p$ given by
\begin{equation}\label{eq:p}
p(x,z) = \tr{\Po(z)S(x)}
\end{equation}
for all $x,z$ (see \cite{FH05}).

The optimal approximate joint instruments $\ii^{(s)}$, $\tilde{\ii}^{(t)}$, $\tilde{\ii}_-^{(t)}$ and $\ii_\Lcorner$ described in Sections \ref{sec:sequential} and \ref{sec:exceptional} are all instances of $W$-covariant instruments. They correspond to the following choice of the respective vector measures $S^{(s)}$, $\tilde{S}^{(t)}$, $\tilde{S}_-^{(t)}$ and $S_\Lcorner$
\begin{align*}
S^{(s)}(x) = {} & \frac{1}{d} \big(a_2(s)\id + db(s)\delta_{x,0}\Po(0)\big) \Qo(0) \big(a_2(s)\id + db(s)\delta_{x,0}\Po(0)\big)\,, \\
\tilde{S}^{(t)}(x) = {} & \frac{1}{d} \big(b(t)\id + da_2(t)\delta_{x,0}\Po(0)\big) \Qo(0) \big(b(t)\id + da_2(t)\delta_{x,0}\Po(0)\big)\,, \\
\tilde{S}_-^{(t)}(x) = {} & \frac{1}{d} \big[b(t)\id - \big(da_2(t)+2b(t)\big)\delta_{x,0}\Po(0)\big] \Qo(-x) \times \\
& \qquad\qquad\qquad\qquad \times \big[b(t)\id - \big(da_2(t)+2b(t)\big)\delta_{x,0}\Po(0)\big]\,, \\
S_\Lcorner(x) = {} & \frac{1}{d+1} \bigg\{\frac{1}{d-2}\delta_{x,0}\bigg[\id-\frac{d}{d-1}\big(\Qo(0)-\Po(0)\big)^2\bigg] \\
& \qquad\qquad\qquad\qquad + \frac{d}{(d-1)^2}(1-\delta_{x,0})\big(\id-\Qo(0)\big)\bigg\}\,.
\end{align*}
Note that all the operators $S^{(s)}(x)$, $\tilde{S}^{(t)}(x)$, $\tilde{S}_-^{(t)}(x)$ and $S_\Lcorner(x)$ are guaranteed to be positive semidefinite, since they are positive scalar multiples of projections. In particular, $\ii_\Lcorner$ is a valid instrument as a consequence of the fact that $S_\Lcorner$ is a valid vector measure.

After this preparation, we are now in position to prove our main theorems. The proof is analogous to the derivation of the compatibility region of the pair of sharp meters $(\Qo,\Po)$, as it can be found in \cite[Proposition 5]{CaHeTo12}.

\begin{proof}[Proof of Theorems \ref{thm:main} and \ref{thm:main_neg}]
The meter $\Qo_s$ is the convolution of the sharp meter $\Qo$ with the probability distribution $q_s$ given by
$$
q_s(z) = s\delta_{z,0} + \frac{1}{d}(1-s)
$$
for all $z$. In a similar way, the channel $I_t$ is the Weyl-covariant channel associated with the probability distribution $p_t$ defined as
$$
p_t(x,z) = t\delta_{x,0}\delta_{z,0} + \frac{1}{d^2}(1-t)\,.
$$
This follows from the fact that all states $\rhoo$ satisfy the equality
$$
\sum_{x,z} W(x,z)^*\rhoo W(x,z) = d\id\,,
$$
as it is easily checked by taking the matrix elements of both sides with respect to the standard basis and using \eqref{eq:ort_rel}. Comparing $q_s$ and $p_t$ with the probability distributions $q$ and $p$ of \eqref{eq:q} and \eqref{eq:p}, we see that the meter $\Qo_s$ and the channel $I_t$ admit a joint $W$-covariant instrument if and only if there exists a $\lc{d}$-valued vector measure $S$ on $\Zb_d$ which satisfies the two equations
\begin{equation}\tag{$\ast$}\label{eq:S_conditions_1}
\begin{aligned}
\tr{\vphantom{\bigg[}\Qo(z)\smash{\sum_x} S(x)} & = s\delta_{z,0} + \frac{1}{d}(1-s) \,,\\
\tr{\Po(z)S(x)} & = t\delta_{x,0}\delta_{z,0} + \frac{1}{d^2}(1-t)
\end{aligned}
\end{equation}
for all $x,z$. In turn, by Lemma \ref{lem:trick}, the existence of such $S$ is equivalent to the compatibility of $\Qo_s$ and $I_t$. In order to prove the theorems, for any fixed value of $t$, we will therefore determine the minimal and the maximal values $s_{\rm min}(t)$, $s_{\rm max}(t)$ of $s$ which allow the existence of a vector measure $S$ satisfying \eqref{eq:S_conditions_1}. Then, the  meter $\Qo_s$ and the channel $I_t$ are compatible for all $s\in [s_{\rm min}(t),s_{\rm max}(t)]$ by an easy convexity argument. To find $s_{\rm min}(t)$ and $s_{\rm max}(t)$, we fix $S$ as in \eqref{eq:S_conditions_1} and, for all $x$, we let $\xi(x)\in\Cb^d\otimes\Cb^d$ be a purification of $S(x)$, i.e., $S(x)=\ptr{\bar{1}}{\kb{\xi(x)}{\xi(x)}}$. Equations \eqref{eq:S_conditions_1} then rewrite
\begin{equation*}
\begin{aligned}
\sum_x \ip{\xi(x)}{(\Qo(z)\otimes\id)\xi(x)} & = s\delta_{z,0} + \frac{1}{d}(1-s) \,,\\
\ip{\xi(x)}{(\Po(z)\otimes\id)\xi(x)} & = t\delta_{x,0}\delta_{z,0} + \frac{1}{d^2}(1-t)\,.
\end{aligned}
\end{equation*}
By decomposing $\xi(x)$ as
$$
\xi(x) = \sum_z \psi_z\otimes\xi_z(x)\,,
$$
where $\xi_0(x),\ldots,\xi_{d-1}(x)$ are vectors of $\Cb^d$, the second equation becomes
$$
\no{\xi_z(x)}^2 = t\delta_{x,0}\delta_{z,0} + \frac{1}{d^2}(1-t) = \left(a_2(t)\delta_{x,0}\delta_{z,0} + \frac{1}{d} b(t)\right)^2
$$
for all $x,z$. Thus,
$$
\xi_z(x) = \left(a_2(t)\delta_{x,0}\delta_{z,0} + \frac{1}{d} b(t)\right) \eta_z(x)\,,
$$
where $\eta_z(x)$ is a unit vector of $\Cb^d$. By using the relation $\ip{\phii_u}{\psi_v} = \omega^{uv}/\sqrt{d}$, the first equation then yields
\begin{align*}
s\delta_{z,0} + \frac{1}{d}(1-s) & = \sum_{x,u,v} \ip{\phii_z}{\psi_u} \ip{\psi_v}{\phii_z} \ip{\xi_v(x)}{\xi_u(x)} \\
& = \frac{1}{d} \sum_x \bigg\|\sum_u \omega^{uz} \xi_u(x)\bigg\|^2 \\
& = \frac{1}{d} \sum_x \bigg\|\sum_u \omega^{uz} \left(a_2(t)\delta_{x,0}\delta_{u,0} + \frac{1}{d} b(t)\right) \eta_u(x)\bigg\|^2\,,
\end{align*}
which gives
$$
s = \frac{1}{d-1} \bigg[ \sum_x \bigg\|\sum_u \left(a_2(t)\delta_{x,0}\delta_{u,0} + \frac{1}{d} b(t)\right) \eta_u(x)\bigg\|^2 - 1 \bigg]
$$
in the particular case with $z=0$. By triangular inequality,
\begin{align*}
\bigg\|\sum_u \left(a_2(t)\delta_{x,0}\delta_{u,0} + \frac{1}{d} b(t)\right) \eta_u(x)\bigg\| & \leq a_2(t)\delta_{x,0} + b(t)
\end{align*}
and therefore
\begin{align*}
s & \leq \frac{1}{d-1} \big[ \big(a_2(t) + b(t)\big)^2 + (d-1)b(t)^2 - 1 \big] \\
& = \frac{1}{d^2}\left[(d^2 - 2) (1 - t) + 2\sqrt{(1 - d^2) t^2 + (d^2 - 2) t + 1}\right]\,.
\end{align*}
For $s$ attaining the above equality, a joint instrument for $\Qo_s$ and $I_t$ actually exists, namely, the instrument $\tilde{\ii}^{(t)}$. It follows that
$$
s_{\rm max}(t) = \frac{1}{d^2}\left[(d^2 - 2) (1 - t) + 2\sqrt{(1 - d^2) t^2 + (d^2 - 2) t + 1}\right]\,.
$$
This proves the second inequality of \eqref{eq:QI_comp_neg}. On the other hand, we have
\begin{align*}
& \sum_x \bigg\|\sum_u \left(a_2(t)\delta_{x,0}\delta_{u,0} + \frac{1}{d} b(t)\right) \eta_u(x)\bigg\|^2 \geq \\
& \qquad\qquad \geq \bigg\|\sum_u \left(a_2(t)\delta_{u,0} + \frac{1}{d} b(t)\right) \eta_u(0)\bigg\|^2 \geq \\
& \qquad\qquad \geq \bigg[\left(a_2(t) + \frac{1}{d} b(t)\right) - \frac{1}{d} b(t) \bigg\|\sum_u (1-\delta_{u,0}) \eta_u(0)\bigg\|\bigg]^2 \,,
\end{align*}
where the last relation follows from another application of the triangular inequality. By the same reason,
$$
\bigg\|\sum_u (1-\delta_{u,0}) \eta_u(0)\bigg\| \leq d-1 \,.
$$
If $t\geq (d-2)/[2(d-1)]$, we have
$$
a_2(t) + \frac{1}{d} b(t) \geq \frac{d-1}{d} b(t) \,,
$$
which implies
\begin{align*}
& \bigg[\left(a_2(t) + \frac{1}{d} b(t)\right) - \frac{1}{d} b(t) \bigg\|\sum_u (1-\delta_{u,0}) \eta_u(0)\bigg\|\bigg]^2  \\
& \qquad\qquad\qquad\qquad\qquad\qquad \geq \bigg[\left(a_2(t) + \frac{1}{d} b(t)\right) - \frac{d-1}{d} b(t) \bigg]^2 \,.
\end{align*}
In the case with $d=2$ the last inequality is true also for $t<(d-2)/[2(d-1)]$ and it is actually an equality. It follows that, for $d=2$ or $t\geq (d-2)/[2(d-1)]$,
\begin{align*}
s & \geq \frac{1}{d-1} \bigg\{ \bigg[\left(a_2(t) + \frac{1}{d} b(t)\right) - \frac{d-1}{d} b(t) \bigg]^2 - 1 \bigg] \\
& = - \frac{2}{d^2}\left[1 - t + \sqrt{(1 - d^2) t^2 + (d^2 - 2) t + 1}\right]\,.
\end{align*}
We can find a joint instrument for $\Qo_s$ and $I_t$ also when $s$ attains the last equality, namely, the instrument $\tilde{\ii}_-^{(t)}$. As a consequence, if we still assume $d=2$ or $t\geq (d-2)/[2(d-1)]$, then
$$
s_{\rm min}(t) = - \frac{2}{d^2}\left[1 - t + \sqrt{(1 - d^2) t^2 + (d^2 - 2) t + 1}\right] \,.
$$
In particular, $s_{\rm min}\big((d-2)/[2(d-1)]\big) = m_1$. When $d\geq 3$, we have also $s_{\rm min}(m_2) = m_1$ because $\Qo_{m_1}$ and $I_{m_2}$ admit the joint instrument $\ii_\Lcorner$, hence $s_{\rm min}(t) = m_1$ for all $t\in\big[m_2,(d-2)/[2(d-1)]\big]$ by the usual convexity argument. The first inequality of \eqref{eq:QI_comp_neg} then follows by combining the two cases with $d=2$ and $d\geq 3$. 
\end{proof}
With the notations of the above proof, the noise parameter $s$ attains its maximal value if and only if there exists a unit vector $\eta(x)$ such that $\eta_u(x) = \eta(x)$ for all $x,u$. 
Indeed, there is no other possibility of saturating the triangular inequality that we used to derive $s_{\rm max}(t)$. 
If $\eta_u(x) = \eta(x)$ for all $x,u$, then $\ptr{\bar{1}}{\kb{\xi(x)}{\xi(x)}} = \tilde{S}^{(t)}(x)$ for all $x$. 
We conclude that, when the point $(s,t)$ attains the second equality in \eqref{eq:QI_comp_neg}, the instrument $\tilde{\ii}^{(t)}$ is actually the unique $W$-covariant joint instrument for $\Qo_s$ and $I_t$. It is an open question whether in this case there exist other (necessarily not $W$-covariant) joint instruments. 
Remarkably, for the joint instruments $\tilde{\ii}^{(t)}_-$ and $\ii_\Lcorner$, the uniqueness property is unclear also if we restrict to $W$-covariant joint instruments.

\section{Concluding remarks}\label{sec:conc}

We have compared three different ways to perform an approximate joint measurement of two incompatible meters: by directly measuring an approximate joint meter, by implementing a sequential measurement scheme, and by making use of approximate cloning. In this chain of scenarios it is allowed to postpone the choice of an increasing number meters: none, one, and two, respectively. One would quite naturally expect that the larger is the level of choice, the worse becomes the quality of the approximation. However, we have shown that the question is more intricate. Indeed, we proved that the first two scenarios are equivalent if we restrict to target meters that are sharp and mutually unbiased, while the last two scenarios are equivalent if we only assume that the target meters are sharp. Remarkably, of these equivalences, only the first one extends to the case of overnoisy measurements, i.e., noisy measurements in which the intensity of the noise exceeds unity. All our results rely on a detailed analysis of the compatibility regions involved in the three scenarios, whose structure reveals a striking similarity for devices of different types (pairs of meters, pairs of channels or mixed meter-channel pairs). We have found that this similarity is broken only by overnoisy devices. Moreover, the qubit system emerges as a very special case when the noise intensity can exceed unity.

The target meters that we considered throughout this paper were always sharp and the noise model always consisted of the uniform noise. It is natural to ask whether our results extend also to unsharp meters, or, if not, for which class of unsharp meters they still hold true. Moreover, when dealing with unsharp meters, other noise models are available besides the uniform one \cite{DeFaKa19}. It would be interesting to see whether the equivalences of the above scenarios depend or not on the considered noise model. Finally, further perspective may emerge by considering our results within the framework of compatibility witnesses. Indeed, the compatibility region of two mutually unbiased sharp meters was derived in \cite{CaHeTo19} by constructing a suitable family of incompatibility witnesses and characterizing the incompatible pair of meters that are detected by each witness in the family. For the compatibility region of two copies of the identity channel, a similar approach was used in \cite{CaHeMiTo19} only in the particular case of equally noisy devices. On the other hand, no connection between compatibility regions and incompatibility witnesses has still been explored in the mixed meter-channel case. Establishing such connection should most likely give more insight on the uniqueness of the optimal joint devices described in Section \ref{sec:ext2}, with implications for the possibility of determining them only from their action on (quantum or classical) subsystems.

\section{Acknowledgment}

C.C. acknowledges that the research activities  have been carried out in the framework of the INFN Research Project QGSKY.

T.H. acknowledges the financial support from the Business Finland project BEQAH.

A.T.~has been supported by the MUR grant Dipartimento di Eccellenza 2023–2027 of Dipartimento di Matematica, Politecnico di Milano, and by the INdAM -- GNAMPA project Probabilit\`a quantistica e applicazioni (CUP E53C23001670001).

\newpage

\appendix

\section{Channels preserving uniform noise}\label{app:noise}

The present appendix contains a detailed proof of the next two results that we used in Sections \ref{sec:joint} and \ref{sec:sequential}.
\begin{proposition}\label{prop:app_noise_1}
Suppose $\Gamma:\lc{d}\to\elle{\Cb^d\otimes\Cb^d}$ is a channel. Then, the following two conditions are equivalent:
\begin{enumerate}[(i)]
\item given any two sharp meters $\Qo$ and $\Po$, the meter $\Go^\Gamma_{\Qo,\Po}$ defined by the relation $\Go^\Gamma_{\Qo,\Po}(x,y) = \Gamma^\dag(\Qo(x)\otimes\Po(y))$ is an approximate joint meter for $\Qo$ and $\Po$;\label{it:app_noise_1_1}
\item the compositions $\ptr{\bar{1}}\circ\Gamma$ and $\ptr{\bar{2}}\circ\Gamma$ are depolarizing channels.
\end{enumerate}
\end{proposition}
\begin{proposition}\label{prop:app_noise_2}
Suppose $\Qo$ is a sharp meter and $\ii$ is an instrument. Moreover, assume the meter $\ii^\dag(\cdot,\id)$ is a noisy version of $\Qo$. Then, the following two conditions are equivalent:
\begin{enumerate}[(i)]
\item given any sharp meter $\Po$, the meter $\Go^\ii_\Po$ defined by the relation $\Go^\ii_\Po(x,y) = \ii^\dag(x,\Po(y))$ is an approximate joint meter for $\Qo$ and $\Po$;\label{it:app_noise_2_1}
\item the sum $\sum_x\ii(x,\cdot)$ is a depolarizing channel.
\end{enumerate}
\end{proposition}

By setting $\Phi^\Gamma_i = \ptr{\bar{i}}\circ\Gamma$, condition \eqref{it:app_noise_1_1} of Proposition \ref{prop:app_noise_1} is equivalent to require that $(\Phi^\Gamma_1)^\dag\circ\Qo = \Qo_{s(\Qo)}$ and $(\Phi^\Gamma_2)^\dag\circ\Po = \Po_{t(\Po)}$ for all sharp meters $\Qo$ and $\Po$, where the noise parameters $s(\Qo),t(\Po)\in [0,1]$ may depend on $\Qo$ and $\Po$. Indeed, this is directly implied by the relations $\sum_y \Go^\Gamma_{\Qo,\Po}(x,y) = (\Phi^\Gamma_1)^\dag(\Qo(x))$ and $\sum_x \Go^\Gamma_{\Qo,\Po}(x,y) = (\Phi^\Gamma_2)^\dag(\Po(y))$. In a similar way, if we set $\Phi^\ii = \sum_x\ii(x,\cdot)$, as a consequence of the relation $\sum_x \Go^\ii_\Po(x,y) = (\Phi^\ii)^\dag(\Po(y))$, condition \eqref{it:app_noise_2_1} of Proposition \ref{prop:app_noise_2} is equivalent to require that $(\Phi^\ii)^\dag\circ\Po = \Po_{t(\Po)}$ for all sharp meters $\Po$, where the noise parameter $t(\Po)\in [0,1]$ may again be $\Po$-dependent. In order to remove the dependence on sharp meters from the noise parameters, we then need the next general result, which implies Propositions \ref{prop:app_noise_1} and \ref{prop:app_noise_2} as immediate corollaries.

\begin{proposition}\label{prop:lin}
Suppose the linear map $\Phi:\lc{d}\to\lc{d}$ is such that its adjoint $\Phi^\dag$ satisfies the relation $\Phi^\dag(\kb{\xi}{\xi}) = s(\xi) \kb{\xi}{\xi} + (1-s(\xi))\tfrac{1}{d}\id$ for all unit vectors $\xi\in\Cb^d$, where $s(\xi)\in\Rb$ possibly depends on $\xi$. Then, $\Phi = s_0 I + (1-s_0)\tfrac{1}{d}\id {\rm tr}$ for some fixed $s_0\in\Rb$.
\end{proposition}
The proof follows from the simple lemma below.
\begin{lemma}\label{lem:lin}
Suppose $\ca{V}$ is a linear space over the field of scalars $\Fb=\Rb$ or $\Fb=\Cb$, and let $T:\ca{V}\to\ca{V}$ be a linear map such that $T\xi = f(\xi)\xi$ for all nonzero $\xi\in\ca{V}$, where $f(\xi)\in\Fb$ possibly depends on $\xi$. Then, the map $f:\ca{V}\setminus\{0\}\to\Fb$ is constant.
\end{lemma}
\begin{proof}
For any two $\xi_1,\xi_2\in\ca{V}\setminus\{0\}$, either $\xi_1\in\Fb\xi_2$, in which case the equality $f(\xi_1) = f(\xi_2)$ is trivial, or $\xi_1$ and $\xi_2$ are linearly independent. In the latter case the equality $f(\xi_1) = f(\xi_2) = f(\xi_1+\xi_2)$ immediately follows from the relation
$$
f(\xi_1) \xi_1 + f(\xi_2) \xi_2 = T(\xi_1) + T(\xi_2) = T(\xi_1+\xi_2) = f(\xi_1+\xi_2) (\xi_1 + \xi_2)\,.
$$
\end{proof}
\begin{proof}[Proof of Proposition \ref{prop:lin}]
It is enough to prove that $s(\xi) = s_0$ for some $s_0\in\Rb$ and all unit vectors $\xi\in\Cb^d$. Indeed, if this is the case, the resulting relation $\Phi^\dag(A) = s_0 A + (1-s_0)\tfrac{1}{d}\id\tr{A}$ holding for $A=\kb{\xi}{\xi}$ extends to all $A\in\lc{d}$ because $\lc{d}$ is linearly spanned by rank-$1$ operators, and then $\Phi = \Phi^\dag = s_0 I + (1-s_0)\tfrac{1}{d}\id {\rm tr}$. We therefore proceed by picking any two unit vectors $\xi_1,\xi_2\in\Cb^d$ and proving that $s(\xi_1) = s(\xi_2)$. Since the latter equality is trivial if $\kb{\xi_1}{\xi_1} = \kb{\xi_2}{\xi_2}$, we assume that $\kb{\xi_1}{\xi_1} \neq \kb{\xi_2}{\xi_2}$. Then, let $P$ be the projection of $\Cb^d$ onto the $2$-dimensional linear subspace $\ca{V}$ spanned by $\xi_1$ and $\xi_2$. For any $\xi\in\ca{V}\setminus\{0\}$, define the unit vector $\xi^\parallel = \xi / \no{\xi}$ and fix any unit vector $\xi^\perp$ of $\ca{V}$ such that $\ip{\xi}{\xi^\perp} = 0$. By observing that
\begin{align*}
\Phi^\dag(P) & = \Phi^\dag(\kb{\xi^\parallel}{\xi^\parallel} + \kb{\xi^\perp}{\xi^\perp}) \\
& = s(\xi^\parallel) \kb{\xi^\parallel}{\xi^\parallel} + s(\xi^\perp) \kb{\xi^\perp}{\xi^\perp} + ( 2 - s(\xi^\parallel) - s(\xi^\perp) )\tfrac{1}{d}\id
\end{align*}
and denoting
$$
c(\xi) = \left[s(\xi^\parallel) + \tfrac{1}{d} (2-s(\xi^\parallel)-s(\xi^\perp)\right] \,,
$$
we find the relation
$$
\Phi^\dag(P) \xi = c(\xi) \xi \,.
$$
Lemma \ref{lem:lin} then implies that $c(\xi) = c_0$ for some $c_0\in\Rb$ and all $\xi\in\ca{V}\setminus\{0\}$. In particolar, $c(\xi^\parallel) = c(\xi^\perp)$ and then $s(\xi^\parallel) = s(\xi^\perp)$, which yields
$$
c_0 = c(\xi) = \tfrac{1}{d} \left[ 2 + (d-2) s(\xi^\parallel) \right]
$$
for all $\xi\in\ca{V}\setminus\{0\}$. If $d \geq 3$, the last relation applied to $\xi=\xi_1$ and $\xi=\xi_2$ implies that $s(\xi_1) = s(\xi_2)$, as claimed. If $d=2$, we define the real linear space
$$
\lc{2}_0 = \{A\in\lc{2}\mid A^\ast = A \text{ and } \tr{A} = 0\}
$$
and we observe that any vector $\xi\in\Cb^2\setminus\{0\}$ induces an operator $A_\xi\in\lc{2}_0$ defined as
$$
A_\xi = \no{\xi}^2 (\kb{\xi^\parallel}{\xi^\parallel} - \kb{\xi^\perp}{\xi^\perp})\,.
$$
This operator satisfies the equality
\begin{align*}
\Phi^\dag(A_\xi) = {} & \no{\xi}^2 \left[ s(\xi^\parallel) \kb{\xi^\parallel}{\xi^\parallel} + (1-s(\xi^\parallel)) \tfrac{1}{2} \id + \right. \\
{} & \left. \qquad\quad - s(\xi^\perp) \kb{\xi^\perp}{\xi^\perp} - (1-s(\xi^\perp)) \tfrac{1}{2} \id \right] = s(\xi^\parallel) A_\xi
\end{align*}
since $s(\xi^\parallel) = s(\xi^\perp)$. As a consequence of spectral theorem, every nonzero $A\in\lc{2}_0$ is such that $A = A_\xi$ for some $\xi\in\Cb^2\setminus\{0\}$. Then, the mapping $\xi\mapsto s(\xi^\parallel)$ is constant by Lemma \ref{lem:lin}, and in particular $s(\xi_1) = s(\xi_2)$ also in this case.
\end{proof}

\section{The extremal channel $\Gamma_\Lcorner$}\label{app:Gamma_J}

In this appendix we prove that the map $\Gamma_\Lcorner$ defined in \eqref{eq:Gammacorner} is a joint channel for two copies of the channel $I_{m_2}$. In this way, we complete the argument of \cite[Eq.~(24) and Fig.~B.1]{Hashagen17} by showing that the point $(m_2,m_2)$ does actually belong to the set of all single clone fidelities that are attainable within universal quantum cloning.

We only prove that $\Gamma_\Lcorner$ is CP, as the equalities $\ptr{\bar{1}}{\Gamma_\Lcorner} = \ptr{\bar{2}}{\Gamma_\Lcorner} = I_{m_2}$ can be verified by routine calculations. To do it, we introduce the {\em Choi operator} of the map $\Gamma_\Lcorner$ with respect to the orthonormal basis $\{\phii_0,\ldots,\phii_{d-1}\}$, i.e., the following element $\hat{\Gamma}_\Lcorner\in\elle{\Cb^d\otimes\Cb^d\otimes\Cb^d}$
$$
\hat{\Gamma}_\Lcorner = \frac{1}{d}\sum_{x,y}\kb{\phii_x}{\phii_y}\otimes\Gamma_\Lcorner(\kb{\phii_x}{\phii_y})\,.
$$
Denoting by $\omega$ the the maximally entangled state of $\Cb^d\otimes\Cb^d$, i.e.,
$$
\omega = \frac{1}{\sqrt{d}}\sum_x\phii_x\otimes\phii_x\,,
$$
we have $\hat{\Gamma}_\Lcorner = (I\otimes\Gamma_\Lcorner)(\kb{\omega}{\omega})$ and then
$$
\hat{\Gamma}_\Lcorner = \frac{1}{d(d^2-2)}\bigg(\id\otimes\id\otimes\id - \sum_{h=\pm 1} V_hV_h^*\bigg)\,,
$$
where $V_h:\Cb^d\to\Cb^d\otimes\Cb^d\otimes\Cb^d$ is the linear map defined by
$$
V_h\xi = \sqrt{\frac{2d}{d+h}}(\id\otimes S_h)(\omega\otimes\xi)
$$
for all $\xi\in\Cb^d$. An easy calculation yields
$$
\ip{(\id\otimes S_h)(\omega\otimes\phii_x)}{(\id\otimes S_h)(\omega\otimes\phii_y)} = \frac{d+h}{2d}\delta_{x,y}
$$
for all $x,y$. This implies that the linear map $V_h$ is an isometry, hence $V_hV_h^*$ is a projection. Since the projections $S_{+1}$ and $S_{-1}$ have mutually orthogonal ranges, the same fact is true for the projections $V_{+1}V_{+1}^*$ and $V_{-1}V_{-1}^*$. It follows that the Choi operator $\hat{\Gamma}_\Lcorner$ is positive semidefinite, hence $\Gamma_\Lcorner$ is CP.

\end{document}